\documentclass[a4paper,UKenglish,cleveref, autoref, thm-restate]{lipics-v2021}

\pdfoutput=1 
\hideLIPIcs  

\title{Approximation Algorithm for the Min-Cost Bipartite Matching with Penalties} 

\author{Eunjin Oh}{Department of Computer Science and Engineering, Pohang University of Science and Technology (POSTECH), Pohang, Republic of Korea \and \url{https://sites.google.com/view/eunjinoh/}}{eunjin.oh@postech.ac.kr}{https://orcid.org/0000-0003-0798-2580}{}

\author{Seongbin Park}{Department of Computer Science and Engineering, Pohang University of Science and Technology (POSTECH), Pohang, Republic of Korea}{seongbin.park@postech.ac.kr}{https://orcid.org/0009-0000-9018-5798}{}

\author{Chanho Song}{Department of Computer Science and Engineering, Pohang University of Science and Technology (POSTECH), Pohang, Republic of Korea}{sch0622@postech.ac.kr}{https://orcid.org/0009-0001-3522-3517}{Supported by the Basic Science Research Program through the National Research Foundation of Korea (NRF) funded by the Ministry of Education (No.RS-2025-25432435).}

\authorrunning{E. Oh, S. Park, and C. Song} 

\Copyright{Eunjin Oh, Seongbin Park, and Chanho Song} 



\category{} 

\funding{Supported by Institute of Information \& Communications Technology Planning \& Evaluation (IITP) grant funded by the Korea government (MSIT) (No.RS-2024-00440239, Sublinear Scalable Algorithms for Large-Scale Data Analysis) and the National Research Foundation of Korea (NRF) grant funded by the Korea government (MSIT) (No.RS-2024-00358505).}

\acknowledgements{All authors contributed equally to this work. Authors are listed in alphabetical order.}

\nolinenumbers 

\EventEditors{John Q. Open and Joan R. Access}
\EventNoEds{2}
\EventLongTitle{42nd Conference on Very Important Topics (CVIT 2016)}
\EventShortTitle{CVIT 2016}
\EventAcronym{CVIT}
\EventYear{2016}
\EventDate{December 24--27, 2016}
\EventLocation{Little Whinging, United Kingdom}
\EventLogo{}
\SeriesVolume{42}
\ArticleNo{23}

\DeclareMathOperator{\poly}{poly}
\newcommand{\ddim}{\textsf{ddim}}

\newcommand{\sdist}{\mathsf d}

\begin{document}

\maketitle

\begin{abstract}
In this paper, we study the minimum-cost bipartite matching with penalties problem in metric spaces with bounded doubling dimension: Given two disjoint sets $R, B$ in a metric space $\mathcal{M}$ with $|R|+|B|=n$ and a penalty function $p \colon R \cup B \to \mathbb{R}_{\ge 0}$, the goal is to select a set of pairs in $R\times B$ so that every point belongs to at most one pair and the sum of the distances of the selected pairs and the penalties of the points not belonging to any pair is minimized.
While near-linear time approximation algorithms are known for the minimum-cost perfect matching problem in geometric settings, no such algorithm was previously known for the penalty setting.
We present a randomized algorithm that computes a $(1+\varepsilon)$-approximate minimum-cost bipartite matching with penalties in $O(n \poly(\log n, 1/\varepsilon))$ time with high probability.
To the best of our knowledge, this is the first near-linear time approximation algorithm for the problem in the penalty setting.
\end{abstract}

\section{Introduction}
A central task in geometric data analysis is to measure the similarity between two finite point sets $R$ and $B$ in a metric space. 
When the two sets have the same cardinality, a classical approach is to match points of $R$ to points of $B$ one-to-one to minimize 
the sum of pairwise distances over the matched pairs.
The resulting problem is known as the \emph{geometric minimum-cost perfect bipartite matching} problem, 
which has been studied extensively over the past several decades~\cite{DBLP:conf/stoc/AgarwalCRX22,DBLP:conf/compgeom/AgarwalV04,DBLP:conf/soda/Indyk07,DBLP:journals/jacm/RaghvendraA20,DBLP:journals/siamcomp/Vaidya89a,DBLP:conf/soda/VaradarajanA99}.
In particular, 
for two point sets $R$ and $B$ in $\mathbb{R}^{\mathrm{d}}$ with $|R|=|B|=n$, 
a $(1+\varepsilon)$-approximate solution can be found in $O(n\poly(\log n,1/\varepsilon))$ expected time~\cite{DBLP:journals/jacm/RaghvendraA20}, and $n(\varepsilon^{-1}\log n)^{O(\mathrm{d})}$ deterministic time~\cite{DBLP:conf/stoc/AgarwalCRX22}.
Despite this progress, the perfect-matching requirement is not always natural, since real data may be unbalanced, noisy, or contaminated by outliers, which can force undesirable high-cost pairings.
This motivates a more flexible geometric matching model in which the solution may be a \emph{non-perfect} matching.

A natural way is the $k$-partial matching, where the goal is to choose exactly $k$ pairs of minimum total length.
Partial matching models arise in real applications where forcing all objects to be matched is inappropriate, such as computer vision tasks with occlusions or outliers, and transportation and logistics settings with imbalanced supply and demand~\cite{DBLP:conf/iclr/JiangLFWY25,DBLP:conf/aaai/WangB22}.
For the $k$-partial matching problem on the plane, Agarwal et al.~\cite{DBLP:conf/compgeom/AgarwalCX19} gave an exact algorithm with running time $O((n+k^2)\operatorname{polylog} n)$ and a $(1+\varepsilon)$-approximation algorithm with running time $O((n+k\sqrt{k})\operatorname{polylog} (n/\varepsilon))$.
However, the $k$-partial matching controls only the number of selected pairs and does not capture point-specific priorities, such as which points should preferably be matched or can be ignored.
Moreover, when $k$ is close to $n$, the known bounds become larger than those for geometric perfect matching.
Another way to handle unbalanced input is many-to-many matching, which relaxes the one-to-one constraint by allowing a point to be matched to multiple points and has been studied for geometric point sets~\cite{an2025approximation,  DBLP:conf/isaac/BandyapadhyayMS21,DBLP:conf/compgeom/Bandyapadhyay024,DBLP:conf/swat/ParkO26}.
This model captures situations where every point should be covered, but it is not appropriate when the desired output is a one-to-one correspondence.

\medskip 
Motivated by the limitations of $k$-partial matching and many-to-many matching, we study the \emph{minimum-cost bipartite matching with penalties} problem in metric spaces.
In this problem, each point is assigned a penalty, and a feasible solution may either match it to a point of the opposite color or leave it unmatched and pay the penalty.
This model captures point-specific importance and preserves the one-to-one nature of the selected pairs.
More formally, 
we are given two disjoint sets $R$ and $B$ in 
a metric space $(\mathcal{M}, d)$ 
along with a \emph{penalty function} $p: X \rightarrow \mathbb{R}_{\geq 0}$, where $X=R\cup B$. 
A \emph{matching} $M \subseteq R \times B$ is a set of edges such that each point in $X$ is incident to at most one edge in $M$. 
The \emph{cost} of $M$ is defined as $\mathrm{w}(M) = \sum_{(r,b) \in M} d(r, b) + \sum_{v \in U(M)} p(v)$, where $U(M)$ is the set of unmatched points in $X$ with respect to $M$. 
The goal is to find a matching that minimizes its cost.

While the one-to-one matching problem has been studied extensively, its penalty variant has been much less explored. 
From a graph-theoretic viewpoint, the penalty variant is not  different from the non-penalty version: by adding dummy vertices that encode the option of leaving points unmatched, one can reduce the problem to the non-penalty variant.
However, this reduction does not preserve the geometric structure of the input.
Thus, from a geometric viewpoint, the penalty variant is fundamentally different from the non-penalty version.
However, except for restricted metric spaces such as tree metrics~\cite{DBLP:conf/nips/SatoYK20},\footnote{
The problem studied in~\cite{DBLP:conf/nips/SatoYK20} is more general; our problem can be viewed as an integral unit-mass analogue of their problem. 
For details, see Appendix~\ref{sec:penaly_matching_tree}.
} 
no near-linear-time algorithm was previously known for the penalty variant in geometric settings.

\subparagraph*{Our result.}
Our main result is stated in Theorem~\ref{thm:main}.
Our algorithm extends beyond Euclidean spaces to metric spaces with bounded \emph{doubling dimension}.
The \emph{doubling dimension} of a metric space, denoted by $\ddim$, is the smallest value $\lambda$ such that every ball can be covered by $2^\lambda$ balls of half the radius.
It is known that $\mathbb{R}^{\mathrm{d}}$ has doubling dimension $\Theta(\mathrm{d})$~\cite{assouad1983plongements}.
\begin{restatable}{theorem}{main}
\label{thm:main}
Let $(\mathcal{M}, d)$ be a metric space of constant doubling dimension $\ddim$.
Let $R$ and $B$ be two point sets in $\mathcal M$,
and let $X= R \cup B$ with $n=|X|$. 
Let $p \colon X \to \mathbb{R}_{\ge 0}$ be a penalty function.
Given a parameter $\varepsilon > 0$, we can compute a $(1+\varepsilon)$-approximate minimum-cost bipartite matching with penalties between $R$ and $B$ in $O(n \poly(\log n, 1/\varepsilon))$ time, with probability at least $1 - 1 / n^{\Omega(1)}$.
\end{restatable}
Note that setting all penalties sufficiently large yields a near-linear time approximation algorithm for minimum-cost bipartite perfect matching in metrics of bounded doubling dimension, which extends the result of~\cite{DBLP:journals/jacm/RaghvendraA20} for fixed-dimensional Euclidean spaces.

\subparagraph*{Additional related work.}
Minimum-cost bipartite matching is a classical problem in graph theory, but directly applying graph algorithms to geometric instances is expensive because the complete bipartite graph has quadratic size.
This has motivated a long line of work on exploiting geometric structure~\cite{DBLP:conf/stoc/AgarwalCRX22,DBLP:conf/compgeom/AgarwalV04,DBLP:conf/soda/Indyk07,DBLP:journals/jacm/RaghvendraA20,DBLP:journals/siamcomp/Vaidya89a,DBLP:conf/soda/VaradarajanA99}.
For geometric bipartite perfect matching in the plane, the best-known exact algorithm runs in $O(n^{2.5}\operatorname{polylog} n)$ time under the $L_1$, $L_2$, and $L_\infty$ metrics~\cite{DBLP:journals/siamcomp/Vaidya89a}.
If the point sets have bounded integer coordinates in $[\Delta]^2$, the optimal solution can be computed in $O(n^{3/2+\delta}\log(n\Delta))$ time for any fixed $\delta>0$~\cite{DBLP:conf/compgeom/Sharathkumar13}.

The penalty variant studied in this paper can also be viewed as a generalization of many-to-many matching, since the standard many-to-many matching cost is recovered when the penalty of each point is set to its nearest-neighbor distance to the opposite color
\cite{DBLP:conf/compgeom/Bandyapadhyay024}.
Thus we provide an alternative near-linear-time approximation algorithm for the many-to-many matching problem in doubling metric spaces. 
The many-to-many matching can be solved exactly in $O(n\log n)$ time for the 1D Euclidean space~\cite{DBLP:journals/gc/ColanninoDHLMRST07}, and 
exactly in $O(n^2\operatorname{polylog} n)$ time for the 2D Euclidean space.
For planar point sets with bounded integer coordinates in $[\Delta]^2$, the running time can be improved to 
$O(n^{3/2}\log (n\Delta))$ time~\cite{DBLP:conf/swat/ParkO26}.
Also, $O_\varepsilon(n\log n)$-time $(1+\varepsilon)$-approximation algorithms for any doubling space were presented~\cite{an2025approximation,DBLP:conf/compgeom/Bandyapadhyay024}.

Our work is also related to unbalanced optimal transport, where the two input measures need not have the same total mass.
The optimal partial transport problem, in which only a prescribed amount of mass is transported, has been studied from the viewpoint of free boundaries and regularity theory~\cite{caffarelli2010free,figalli2010optimal}.
Piccoli and Rossi~\cite{piccoli2014generalized} studied a generalized Wasserstein distance for different total masses by adding a mass-variation penalty to the transportation cost.
Chizat et al.~\cite{chizat2018unbalanced} later developed a more general unbalanced optimal transport framework where mass variation is modeled by divergence penalties and can be interpreted as mass creation or destruction.
These models are fractional in nature, while our problem is an integral matching problem with arbitrary point-dependent penalties.

\subparagraph*{Terminology for matching.}
Let $H$ be a bipartite graph with edge costs.
A \emph{matching} $M$ in $H$ is a set of pairwise vertex-disjoint edges, and its cost is defined as the sum of its edge costs.
A matching $M$ in $H$ is \emph{perfect} if every vertex is incident to exactly one edge of $M$.
A vertex of the graph is \emph{free} (with respect to $M$) if no edge of $M$ is incident to it. 
An \emph{alternating path} (with respect to $M$) is a path in the graph whose edges alternately belong to $M$ and not to $M$. 
An \emph{augmenting path} (with respect to $M$) is an alternating path whose two endpoints are both free. 
Augmenting $M$ along an augmenting path $P$ by taking the symmetric difference $M \oplus P$ increases its cardinality by one.
A matching $M$ has maximum cardinality if no augmenting path exists.
Classically, maximum-cardinality (or minimum-cost) matchings are computed by repeatedly augmenting $M$ along paths with certain properties until none remain.

\section{Main Obstacles and Our Approach}
A natural first attempt to handle penalties is to incorporate them directly into the augmenting-path framework of Raghvendra and Agarwal~\cite{DBLP:journals/jacm/RaghvendraA20} for minimum-cost perfect bipartite matching.
In this section, we explain why this direct approach fails and how we overcome the difficulty.

\subparagraph*{The framework of Raghvendra and Agarwal.}
The framework of Raghvendra and Agarwal~\cite{DBLP:journals/jacm/RaghvendraA20}
builds on the following classical algorithm of Gabow and Tarjan~\cite{doi:10.1137/0218069}: 
Starting from $M=\emptyset$, the algorithm repeatedly computes a minimum
\emph{net-cost} augmenting path $P$ and augments $M$ along $P$, until $M$
becomes perfect. For now, let us define the \emph{net cost} of an augmenting path
$P$ as
$\mathrm{w}(M\oplus P)-\mathrm{w}(M)+c\cdot |P\setminus M|$ for a fixed constant $c=\Theta(\varepsilon \mathrm{w}(M^*)/n)$.
Thus it favors augmenting paths that are
short while increasing the matching cost by only a small amount. At the end, the
resulting matching is a $(1+\varepsilon)$-approximate perfect matching. Moreover, 
the total length of all augmenting paths is
$O((n/\varepsilon)\log n)$.

However, computing all $n$ augmenting paths in near-linear time is still
nontrivial, even though the number of edges in these paths is small. To achieve a near-linear running time, Raghvendra and
Agarwal~\cite{DBLP:journals/jacm/RaghvendraA20} use a quadtree-based
perturbation of the distances between vertices in $R\cup B$. Also, for technical reasons, Raghvendra and Agarwal refine the third term $c\cdot |P\setminus M|$ 
in the definition of the net cost.
This perturbation
allows them to compute a minimum net-cost augmenting path in time near-linear
in the number of edges of the path.

\subparagraph*{Why the direct approach is insufficient.}
In the penalty setting, an optimal matching may leave some vertices unmatched.
Let us start with the classical algorithm by Gabow and Tarjan.
As before, the net cost of an alternating path $P$
is defined 
as $\mathrm{w}(M\oplus P)-\mathrm{w}(M) + c\cdot |P\setminus M|$. Here, $\mathrm{w}(\cdot)$ is
the total length of all matched edges plus the total penalty of all unmatched vertices. 
The algorithm repeatedly augments the current matching along a minimum 
net-cost augmenting path. This produces a sequence of partial matchings, one for
each iteration, and one can show that the minimum-cost partial matching among
those computed is a $(1+\varepsilon)$-approximate matching with penalties.
Moreover, by slightly modifying the data structure of 
Raghvendra and Agarwal, one can compute a minimum net-cost augmenting path in time near-linear in the number of edges in the path (after the weight perturbation).

However, it is unclear whether the total length of all augmenting paths computed
in this way is near-linear. More precisely, the analysis of Gabow and
Tarjan~\cite{doi:10.1137/0218069} does not carry over directly. Let $M^*$ be an
optimal matching with penalties. Let $M$ be the current matching. 
Then $M^* \oplus M$
decomposes into alternating cycles and alternating paths; not all these paths 
are augmenting paths with respect to $M$.
Moreover, 
the number of augmenting paths having both endpoints unmatched in $M$ can be smaller than $n-i$
even if $|M|=i$. 
This is precisely where the penalty setting differs from the non-penalty variant, and 
these properties are crucial in the analysis  in~\cite{doi:10.1137/0218069}.
In the non-penalty variant, 
the net cost of a minimum net-cost augmenting path is at most the average net cost of these $n-i$ augmenting paths,
namely
$\frac{\mathrm{w}(M^*)-\mathrm{w}(M)+cn}{n-i}$.
By summing this bound over all iterations, together with
the lower bound in which the perturbation term charges the length of the chosen
paths, we can show that the total length of all augmenting paths is 
$O((n/\varepsilon)\log n)$.
However, this analysis does not carry over directly to the penalty variant due to the differences we mentioned earlier.

\subparagraph*{Our approach.}
Our algorithm keeps the spirit of the Gabow–Tarjan framework,
but allows some \emph{alternating} paths to change the current
matching. In the non-penalty setting, every useful augmentation connects two
free vertices and increases the cardinality of the matching. 
In the penalty
setting, however, this is too restrictive: an optimal solution may prefer to
replace the currently unmatched vertex rather than increase the cardinality.
Thus, besides ordinary augmenting paths whose two endpoints are free, we also
need to consider alternating paths that start at a free vertex and end at a
matched vertex. Augmenting along such a path preserves the cardinality, but it
changes which vertex remains unmatched and hence changes the penalty paid by the
solution.
To make these two operations comparable, we define the net cost so that it
accounts not only for the change in matched-edge length, but also for the change
in unmatched-vertex penalties. The definition is chosen so that ordinary
augmenting paths are favored when they give comparable cost, while exchange-type
alternating paths are selected when the saving in penalties justifies them. With
this choice, the same charging intuition as in the Gabow–Tarjan analysis can be
recovered: the perturbation term controls the total length of the paths, and the
net-cost optimality gives the approximation guarantee.

To formulate this idea, we use 
the \emph{prism
graph} where the vertex set is divided into two \emph{layers}, and show that 
it provides the right way to view this idea. 
Both operations become ordinary augmentations with augmenting paths in the prism graph. A path contained in one layer corresponds to an augmenting path in the original instance, whereas a path that
crosses the two layers corresponds to an alternating path that starts at a free
vertex. 
Thus, the prism graph turns the above
informal rule into a standard framework.
This interpretation is conceptually simple, but making it algorithmic requires
several technical ingredients; developing these ingredients is the main technical
contribution of this paper.

\section{Main Algorithm}
In this section, we first introduce the prism graph, and
provide a reduction from the original problem into 
the minimum cost perfect matching problem on the prism graph.
Then we show that the framework of~\cite{DBLP:journals/jacm/RaghvendraA20}, which gives a near-linear time algorithm for a minimum cost perfect matching problem in \emph{Euclidean spaces}, can be applied in the prism graph.  
Later, we will see that the running time of the entire algorithm
depends on the 
height of the split-tree of $X$.
However, if the spread of $X$, that is, the ratio of the longest distance and the smallest distance, is unbounded,
the height of the split-tree is unbounded.
Thus as preparation, we reduce the problem into an instance with bounded spread.

\subparagraph*{Decomposition into subproblems with bounded spread.}
We decompose the problem into independent subproblems defined on pseudometric spaces with bounded spread and bounded doubling dimension.
We first scale the metric and the penalty function to bound the optimal cost, and map the points to representative centers to enforce a minimum non-zero distance of at least one.
We then apply a randomized low-diameter decomposition to partition the space into disjoint clusters of bounded \emph{diameter}\footnote{The \emph{diameter} of a set is the maximum distance between any pair of points.}, preserving an optimal matching with probability at least $1/2$.
$(1+\varepsilon)$-approximate solutions of independent subproblems yield a $(1+O(\varepsilon))$-approximate solution of the original instance.
We then focus on a single cluster and redefine the problem instance under the pseudometric and the scaled penalty function.
Distinct points may have a distance of zero under this pseudometric, and we maintain these co-located points as distinct entities to preserve their individual penalties.
Consequently, we assume the input instance has a minimum non-zero distance of at least one and a bounded spread of $O(n^2/\varepsilon)$.
This bounded spread restricts the depth of the split-tree of the resulting point set and enables a near-linear running time.
Details are provided in Section~\ref{sec:decompose_bounded_spread}.

\subsection{Prism Graph}\label{sec:prism}
We now describe the prism graph,
and provide a reduction from the original problem into 
the \emph{minimum-cost perfect matching} problem on the prism graph.
The \emph{prism graph} $\tilde{G}=(V_0\cup V_1,\tilde{E})$ is defined as follows. 
We create two layers $V_0$ and $V_1$. 
The upper layer $V_0=R_0\cup B_0$ is an ordinary copy of the input, where $R_0$ and $B_0$ are copies of $R$ and $B$, respectively. 
The lower layer $V_1=R_1\cup B_1$ is a color-reversed copy, where $R_1$ contains copies of points in $B$ and $B_1$ contains copies of points in $R$. 
For each point $u\in R\cup B$, we use $u_0$ and $u_1$ to denote its copies in the upper and lower layers, respectively. 
The edge set $\tilde{E}$ consists of three disjoint types of edges, each with its corresponding cost $\tilde{c}(\cdot,\cdot)$:

\begin{itemize}
    \item For each pair $(r, b) \in R \times B$, we add an edge $(r_0, b_0) \in R_0 \times B_0$ with cost $\tilde{c}(r_0, b_0) =  {d(r, b)}$.
    \item For each pair $(r, b) \in R \times B$, we add an edge $(b_1, r_1) \in R_1 \times B_1$ with cost $\tilde{c}(b_1, r_1) = d(r, b)$.
    \item For each point $u \in R \cup B$, we add an edge connecting its two copies. 
    Specifically, for each $r \in R$, we add $(r_0, r_1) \in R_0 \times B_1$ with cost $\tilde{c}(r_0, r_1) = 2p(r)$. 
    For each $b \in B$, we add $(b_1, b_0) \in R_1 \times B_0$ with cost $\tilde{c}(b_1, b_0) = 2p(b)$.
\end{itemize}

\begin{figure}
    \centering
    \includegraphics[width=0.78\textwidth]{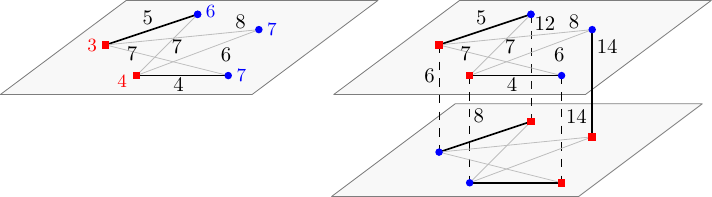}
    \caption{Illustration of the prism graph (right) for the original instance with penalties (left). The left figure shows the min-cost matching with penalties, where the unmatched vertex incurs a penalty. Note that the colored values next to vertices denote penalties. The right figure shows the corresponding min-cost perfect matching in the prism graph. Thick segments indicate edges in the optimal matching in the original instance and the prism graph.}
    \label{fig:prism_graph}
\end{figure}

We call the edges inside a single layer the \emph{geometric edges}, and the edges connecting two copies of the same point the \emph{link edges}.
The following lemma formalizes the role of the two-layer prism graph and the link edges of cost $2p(v)$.
It shows that the minimum-cost perfect matching value in $\tilde{G}$ is exactly twice the optimum value of the penalty matching problem on $R$ and $B$.
See also Figure~\ref{fig:prism_graph}.

\begin{restatable}{lemma}{reduction} \label{lem:reduction-pm-mcpm}
    Let $\textsf{OPT}$ denote the optimum value of the penalty matching instance on $R$ and $B$, and let $\widetilde{\textsf{OPT}}$ denote the optimum value of a minimum-cost perfect matching in $\tilde{G}$.
    Then, $\widetilde{\textsf{OPT}} = 2 \cdot \textsf{OPT}$. 
    Furthermore, an optimal penalty matching of $R$ and $B$ can be recovered from an optimal perfect matching in $\tilde{G}$ in $O(n)$ time.
\end{restatable}
\begin{proof}
    We prove the equivalence by establishing bounds in both directions.

    Let $M^* \subseteq R \times B$ be an optimal penalty matching. 
    We construct a perfect matching $\tilde{M}$ in $\tilde{G}$ as follows.
    For every matched pair $(r, b) \in M^*$, we add both $(r_0, b_0)$ and $(b_1, r_1)$ to $\tilde{M}$.
    For every unmatched point $u \in U(M^*)$, we add its corresponding penalty edge $(u_0, u_1)$ to $\tilde{M}$.
    Because every original point $u \in R \cup B$ is either matched in $M^*$ or left unmatched, exactly one of these cases applies.
    Thus, every vertex in $V_0 \cup V_1$ is incident to exactly one edge in $\tilde{M}$.
    Let $\tilde{c}(\tilde{M})$ be the total cost of $\tilde{M}$.
    Then,
    \begin{align*}
       \tilde{c}(\tilde{M}) &= \sum_{(r,b) \in M^*} \big(\tilde{c}(r_0, b_0) + \tilde{c}(b_1, r_1)\big) + \sum_{u \in U(M^*)} \tilde{c}(u_0, u_1) \\
        &= 2 \sum_{(r,b) \in M^*} d(r,b) + \sum_{u \in U(M^*)} 2p(u) \\
        &= 2 \cdot \textsf{OPT}(R, B).
    \end{align*}
    Since $\widetilde{\textsf{OPT}}$ is the minimum over all perfect matchings, we obtain $\widetilde{\textsf{OPT}} \le 2 \cdot \textsf{OPT}(R, B)$.

    Let $\tilde{M}^*$ be an optimal perfect matching in $\widetilde{G}$ with cost $\widetilde{\textsf{OPT}}$.
    Because the only edges connecting the upper and lower parts of $\tilde{G}$ are the penalty edges, the matching decomposes into three disjoint sets: edges in upper part $\tilde{M}^*_0$, edges in lower part $\tilde{M}^*_1$, and a set of penalty edges $\tilde{M}^*_P$.
    Let $U \subseteq R \cup B$ be the set of original points whose corresponding penalty edges are in $\tilde{M}^*_P$.
    Since $\tilde{M}^*$ perfectly matches all vertices in $\tilde{G}$, the remaining vertices in the upper part of $\tilde{G}$, which correspond to $(R \cup B) \setminus U$, must be perfectly matched by $\tilde{M}^*_0$.
    Projecting $\tilde{M}^*_0$ and $\tilde{M}^*_1$ back to the original point sets gives two valid penalty matchings $M_0$ and $M_1$, respectively, both of which leave exactly the points in $U$ unmatched.
    The total cost of $\tilde{M}^*$ can be written as
    \begin{align*}
        \widetilde{\textsf{OPT}} &= \sum_{(r_0, b_0) \in \tilde{M}^*_0} d(r_0, b_0) + \sum_{(b_1, r_1) \in \tilde{M}^*_1} d(b_1, r_1) + \sum_{u \in U} 2p(u) \\
        &= \left( \sum_{(r,b) \in M_0} d(r,b) + \sum_{u \in U} p(u) \right) + \left( \sum_{(r,b) \in M_1} d(r,b) + \sum_{u \in U} p(u) \right).
    \end{align*}
    The two parenthetical terms are exactly the objective values of the penalty matchings $M_0$ and $M_1$.
    Since each parenthetical term is the cost of a valid penalty matching, each is at least $\textsf{OPT}$.
    Therefore, $\widetilde{\textsf{OPT}}\ge 2\textsf{OPT}$.

    Combining both bounds establishes $\widetilde{\textsf{OPT}}= 2 \cdot \textsf{OPT}$. 
    This equality also forces the cost of $M_0$ to be exactly $\textsf{OPT}$. 
    Thus, an optimal penalty matching can be trivially recovered in $O(n)$ time by extracting the edges of the upper part of $\tilde{G}$ from $\tilde{M}^*$.
\end{proof}

\subsection{Approximating Distances by the Split-Tree}
In this subsection, we introduce a randomized split-tree and use it to define a split-tree-induced distance function $\sdist$ that approximates $d$.
The randomized split-tree plays the same role in doubling metrics as the randomly shifted quadtree does in Euclidean space.
We use $\sdist$ only for the geometric edges of the prism graph, while the link edges retain their penalty costs.

\subparagraph*{Randomized split-tree.}
Recall that the input point set $X$ has spread $\Delta=O(n^2/\varepsilon)$, and the minimum nonzero distance in $X$ is at least $1$, and set $\delta := \lceil \log_2 \Delta \rceil$.
A \emph{randomized split-tree} on $X$ is a hierarchical partition
$\mathcal{P}_0,\mathcal{P}_1,\ldots,\mathcal{P}_{\delta+1}$, where $\mathcal{P}_0$ is the singleton partition and $\mathcal{P}_{\delta+1}=\{X\}$.
We call each element of $\mathcal{P}_i$ a level-$i$ cluster.
The hierarchical partition is \emph{nested}: Each cluster in $\mathcal{P}_i$ is the union of clusters in $\mathcal{P}_{i-1}$.
Thus it can be viewed as a rooted tree whose leaves are the points of $X$.
Every level-$i$ cluster in $\mathcal{P}_i$ has diameter 
at most $2^{i+1}$, and has at most $2^{O(\ddim)}$ children.
Moreover, for some constant $C_{\mathrm{split}}>0$ depending only on $\ddim$, every pair $u,v\in X$ and every level $i$ satisfy
\begin{align}\label{eq:split-tree-separation}
    \Pr\bigl[u \text{ and } v \text{ belong to different clusters of } \mathcal{P}_i\bigr]
    \le C_{\mathrm{split}} \cdot \frac{d(u,v)}{2^i}.
\end{align}
It says that two points whose distance is much smaller than the scale $2^i$ are unlikely to be separated at level $i$.
We use this property to define an approximate distance induced by the split-tree, which has small expected distortion.

In general metrics, constructing such a randomized split-tree may require quadratic time, since the input metric can contain $\Theta(n^2)$ pairwise distances.
For bounded doubling metrics, however, Cohen-Addad et al.~\cite{DBLP:journals/jacm/Cohen-AddadFS21} showed that a randomized split-tree can be constructed in near-linear time.
In particular, for constant doubling dimension and bounded spread $\Delta=O(n^2/\varepsilon)$,
the construction time is $O(n\log \Delta)=O(n\log (n/\varepsilon))$.

\subparagraph*{Split-tree-induced distance.}
We now define an approximate distance function $\sdist$ using the randomized split-tree.
This distance function plays the same role as the quadtree-based distance function in~\cite{DBLP:journals/jacm/RaghvendraA20}.
The purpose of this distance is to group many geometric edges into a bounded number of classes with identical costs. In other words, this allows us to represent $\Theta(n^2)$ pairwise distances using a near-linear number of classes. 
This is crucial for keeping the complexity of the data structure in Section~\ref{sec:Data-Structure} within our target bound.
We show that the expected error introduced by this perturbation is sufficiently small.

We first construct the split-tree on the set of distinct locations of $X$.
For each cluster $A$ in the split-tree, we fix an arbitrary point $c(A)\in A$ and call it the \emph{representative point} of $A$.
Let $\omega=\left\lceil\log_2\left(\frac{C\delta}{\varepsilon}\right)\right\rceil$, where $C>0$ is a sufficiently large constant depending only on $\ddim$.
For a cluster $A$ in the split-tree, let $\mathcal D[A]$ be the family of descendants of $A$ that lie $\omega$ levels below $A$.
If fewer than $\omega$ levels remain below $A$, we use the  descendants of $A$ in $\mathcal P_0$, which are singletons.
Note that these descendants partition $A$.
We call the elements of $\mathcal D[A]$ the \emph{refined descendants} of $A$.
Since every split-tree cluster has at most $2^{O(\ddim)}$ children, we have $|\mathcal D[A]|=\poly(\log\Delta,1/\varepsilon)$ for constant $\ddim$.

For two points $u,v\in X$ that have the same location, we define $\sdist(u,v):=0$.
For two points $u,v\in X$ with distinct locations, let $\mathcal{L}(u,v)$ be the least common ancestor cluster of $u$ and $v$ in the split-tree, and let $\ell(u,v)$ be its level.
Let $U,V\in\mathcal D[\mathcal{L}(u,v)]$ be the refined descendants containing $u$ and $v$, respectively.
We define the split-tree-induced distance on $X$ by
$\sdist(u,v)
:=
d(c(U),c(V))
+
2^{\ell(u,v)-\omega+2}$. 
The following lemma shows that $\sdist$ dominates the original distance $d$ and has small expected distortion.
\begin{restatable}{lemma}{splittree}
\label{lem:split-tree-distance-approx}
For any $u,v\in X$, we have $d(u,v)\le \sdist(u,v)$ and
 $\mathbb{E}[\sdist(u,v)]\le (1+\varepsilon/2)d(u,v)$.
\end{restatable}
\begin{proof}
If $u$ and $v$ have the same location, then $d(u,v)=0$ and $\sdist(u,v)=0$ by definition.
Thus, both inequalities hold in this case.
Assume from now on that $u$ and $v$ have distinct locations.
Let $U,V\in\mathcal{D}[\mathcal{L}(u,v)]$ be the refined descendants containing the locations of $u$ and $v$, respectively.
Since every level-$i$ cluster has diameter at most $2^{i+1}$ and singleton clusters have diameter $0$, we have
$d(u,c(U))\le 2^{\ell(u,v)-\omega+1}
\quad\text{and}\quad
d(v,c(V))\le 2^{\ell(u,v)-\omega+1}.$
By the triangle inequality,
\begin{align*}
d(u,v)
&\le d(c(U),c(V))+d(u,c(U))+d(v,c(V))
\le d(c(U),c(V))+2^{\ell(u,v)-\omega+2}
= \sdist(u,v).
\end{align*}
This proves the domination property.
We next prove the expected distortion bound.
By the triangle inequality and the same diameter bound,
\begin{align*}
d(c(U),c(V))
&\le d(u,v)+d(u,c(U))+d(v,c(V))
\le d(u,v)+2^{\ell(u,v)-\omega+2}.
\end{align*}
Therefore,
$\sdist(u,v)
=d(c(U),c(V))+2^{\ell(u,v)-\omega+2}
\le d(u,v)+2^{\ell(u,v)-\omega+3}.$

If $\ell(u,v)=i$, then the locations of $u$ and $v$ are separated at level $i-1$.
By the split-tree separation property~(\ref{eq:split-tree-separation}) and using the previous upper bound on $\sdist(u,v)$, we obtain
\begin{align*}
\mathbb{E}[\sdist(u,v)]
&\le d(u,v)+\sum_{i=1}^{\delta+1}\Pr[\ell(u,v)=i]\cdot 2^{i-\omega+3}\\
&\le d(u,v)+\sum_{i=1}^{\delta+1} C_{\mathrm{split}}\cdot \frac{d(u,v)}{2^{i-1}}\cdot 2^{i-\omega+3}\\
&\le d(u,v)+16C_{\mathrm{split}}(\delta+1)2^{-\omega}d(u,v).
\end{align*}
Since $\omega=\lceil\log_2(C\delta/\varepsilon)\rceil$ and $C>0$ is chosen sufficiently large depending only on $\ddim$, we have
$16C_{\mathrm{split}}\cdot\delta \cdot2^{-\omega}\le \varepsilon/2$.
Therefore,
$\mathbb{E}[\sdist(u,v)]\le (1+\varepsilon/2)d(u,v)$.
\end{proof}

In summary, the split-tree-induced distance gives 
the identical cost to geometric edges whose endpoints lie in the same pair of refined descendants.
This uniformity enables a compact cluster-level representation of the graph, while the above lemma guarantees that the resulting approximation error remains small.

\subsection{Locality and Compact Augmenting Paths}
We now classify the geometric edges in the prism graph with respect to a current matching.
This classification is based on the refined descendants used in the split-tree-induced distance and will be used to simplify augmenting paths.

By symmetry, it is enough to define the classification in the upper layer.
Fix a matching $M$ on the prism graph $\tilde{G}$.
Let $M^{\mathrm{up}}$ be the set of matched geometric edges of $M$ in the upper layer.
We partition the edges of $M^{\mathrm{up}}$ into equivalence classes as follows.
For each edge in $M^{\mathrm{up}}$, we write it as $(u,v)$ with $u\in R_0$ and $v\in B_0$.
Recall that we define $\mathcal{L}(u,v)$ as the least common ancestor cluster of the locations of $u$ and $v$ in the split-tree.
If $u$ and $v$ have the same location, we define $\mathcal{L}(u,v)$ as the singleton cluster containing that location.
Let $U$ and $V$ be the refined descendants of $\mathcal{L}(u,v)$ that contain $u$ and $v$, respectively.
Two matched geometric edges $(u,v),(u',v')\in M^{\mathrm{up}}$ belong to the same equivalence class if and only if $\mathcal{L}(u,v)=\mathcal{L}(u',v')$, $U=U'$, and $V=V'$, where $U,V$ and $U',V'$ are the refined descendants associated with $(u,v)$ and $(u',v')$, respectively.
Let $\mathcal{K}_M$ denote the resulting partition of the matched geometric edges in $M^{\mathrm{up}}$.
For an equivalence class $K\in\mathcal{K}_M$, let $R_K$ and $B_K$ be the sets of $R_0$- and $B_0$-endpoints of the edges in $K$, respectively.
A geometric edge $(r,b)\in R_0\times B_0$ is called \emph{local} with respect to $M$ if there exists a class $K\in\mathcal{K}_M$ such that $r\in R_K$ and $b\in B_K$.
Otherwise, $(r,b)$ is called \emph{non-local}.
The same classification is applied symmetrically in the lower layer.
See Figure~\ref{fig:local_nonlocal}(a).

\begin{figure}
    \centering
    \includegraphics[width=0.9\textwidth]{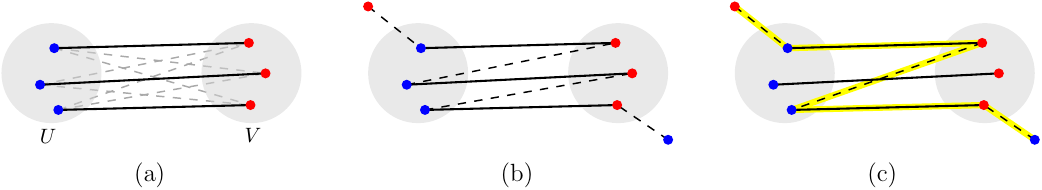}
    \caption{(a) The sets $U$ and $V$ are two refined descendants of the same split-tree cluster. The black edges are the matched geometric edges in the same equivalence class. Hence, every edge between $U$ and $V$ is local.
    (b) An augmenting path containing a long subpath consisting only of local geometric edges. 
    (c) The long local subpath can be replaced by a shorter local subpath with the same net cost. The highlighted augmenting path is compact and has the same net cost as the path in (b).}
    \label{fig:local_nonlocal}
\end{figure}

We define the adjusted cost of an edge $e$ with respect to $M$ by
\begin{align*}
    \Phi_M(e)
    :=
    \begin{cases}
        \sdist(e), & \text{if } e \text{ is a local geometric edge},\\
        \sdist(e)+\theta, & \text{if } e \text{ is a non-local geometric edge},\\
        \tilde{c}(e), & \text{if } e \text{ is a link edge}.
    \end{cases}
\end{align*}
Here, $\theta$ is a parameter such that
$\frac{\varepsilon \tilde{W}^*}{6n}\le \theta \le \frac{\varepsilon \tilde{W}^*}{3n}$, where $\tilde{W}^*$ denotes the optimum perfect matching cost in the prism graph.
Note that link edges do not receive the additional $\theta$ term.
For a link edge $e=(u_0,u_1)$, we have $\Phi_M(e)=\tilde{c}(e)=2p(u)$.

Here, we do not compute $\theta$ exactly.
Instead, we use standard geometric guessing for $\tilde{W}^*$ using the following lemma: for each such value $\theta$, we run the algorithm to compute a matching and return the matching with smallest cost. 
The number of guesses is absorbed into the $\poly(\log n,1/\varepsilon)$ factor in the running time.
\begin{restatable}{lemma}{candidate}\label{lem:candidate}
    We can find $O(\log n)$ values 
    one of which is in
    $[\frac{\varepsilon \tilde{W}^*}{6n}, \frac{\varepsilon \tilde{W}^*}{3n}]$ in near-linear time.
\end{restatable}
\begin{proof}
We first construct a 2-spanner $G_2$ of $(X,d)$, build an HST $H$ from $G_2$, and compute the exact penalty-matching optimum $W_H^*$ under the tree metric $d_H$.
For every matching $M$, the HST guarantee gives
$w(M)\le w_H(M)\le 2(n-1)w(M)$.
Hence, 
$\textsf{OPT}\le W_H^*\le 2(n-1)\textsf{OPT}$,
where $\textsf{OPT}$ denotes the optimum value of the penalty-matching instance under $d$.
Since the optimum value of the prism instance is $\tilde{W}^*=2\textsf{OPT}$, setting
$T:=\frac{W_H^*}{n-1}$
gives
$T\le \tilde{W}^*\le 2(n-1)T.$
Let 
$G_i=2^iT$  for 
$i=0,1,\ldots,\left\lceil\log_2(2(n-1))\right\rceil+1$.
Then we return the set of all values  
$\theta_i=\frac{\varepsilon G_i}{6n}$.
For at least one index $i$, we have
$\tilde{W}^*\le G_i<2\tilde{W}^*.$
Therefore,
$\frac{\varepsilon \tilde{W}^*}{6n}\le \theta_i<\frac{\varepsilon \tilde{W}^*}{3n}$.
\end{proof}

For a set of edges $\Gamma$, let $\Phi_M(\Gamma):=\sum_{e\in \Gamma}\Phi_M(e)$.
For an alternating path or cycle $\Pi$, we define its \emph{net cost} with respect to $M$ as
$\phi_M(\Pi):=\Phi_M(\Pi\setminus M)-\Phi_M(\Pi\cap M)$.

By construction, all local edges in $R_K\times B_K$ have the same adjusted cost.
Indeed, they have the same least common ancestor cluster and the same pair of refined endpoint clusters.
This simple uniformity property is the reason why a long subpath consisting only of local edges can later be replaced by a shorter subpath without changing its net cost.

We call an alternating path $\Pi$ \emph{compact} if the number of local edges in $\Pi$ is at most three times the number of non-local edges in $\Pi$.
Equivalently, if $\Pi_1$ and $\Pi_{2}$ denote the sets of local and non-local edges in $\Pi$, respectively, then
$|\Pi_1|\le 3|\Pi_{2}|$. 
The following lemma shows that it is sufficient to search for compact augmenting paths.
The proof is almost identical to the argument of~\cite{DBLP:journals/jacm/RaghvendraA20}, because compactness is imposed only on geometric edges.
See also Figure~\ref{fig:local_nonlocal}(b--c).

\begin{lemma}[\cite{DBLP:journals/jacm/RaghvendraA20}]
\label{lem:compact-augmenting-path}
Among all augmenting paths of minimum net cost with respect to $M$ in the prism graph $\tilde{G}$, there exists one whose geometric edges are compact.
\end{lemma}

Geometric edges inside the same equivalence class have the same split-tree-induced distance, and hence a long alternating subpath that uses only local geometric edges can be replaced by a shorter subpath without changing its net cost.
Therefore, the essential complexity of an augmenting path is captured by its non-local geometric edges.
To bound the complexity of augmenting paths in terms of their non-local edges, we add the term $\theta$ only to non-local geometric edges in the adjusted cost.
The compactness lemma shows that among minimum-net-cost augmenting paths, it is sufficient to consider paths whose number of local geometric edges is proportional to the number of non-local geometric edges.
This property will later allow us to bound the total length of all augmenting paths.

\subsection{Algorithm}
We now describe a near-linear time approximation algorithm
for computing a minimum cost perfect matching on the prism graph.
Starting from the empty matching, the algorithm repeatedly finds a minimum-net-cost augmenting path with respect to the current matching and augments along it.
Among all minimum-net-cost augmenting paths, we use a minimum-net-cost augmenting path in a specific form, which we call a \emph{symmetric augmenting path}.

Before defining symmetric augmenting paths, we first define the \emph{mirrored path} of a geometric path.
If a geometric path $\pi$ in one layer is mapped to a geometric path in the other layer by replacing each vertex with the copy of the same original point in the other layer, then the resulting path is called the mirrored path of $\pi$ and is denoted by $\widehat{\pi}$.
Next, an augmenting path $\Pi$ is called 
a symmetric augmenting path if it satisfies one of the following two conditions:
(i) $\Pi$ consists only of geometric edges in one layer, or (ii)
$\Pi$ contains exactly one link edge $e$ and can be written as $\Pi=\pi \circ e \circ \hat{\pi}$, where $\pi$ is a geometric subpath in one layer and $\widehat{\pi}$ is its mirrored path in the other layer.
We use symmetric augmenting paths to keep the matching states in the two layers identical throughout the algorithm.

For a symmetric augmenting path belonging to type (i),
we prove below that, after augmenting along $\Pi$, the mirrored path $\hat{\Pi}$ remains a minimum-net-cost augmenting path.
This is not straightforward; $M$ is augmented along $\Pi$, and
the adjusted cost of an edge is defined with respect to the locality of the edges.
Since we show that $\hat{\Pi}$ has minimum net-cost, 
the algorithm is allowed to augment the current matching along $\hat{\Pi}$, so the two augmentations together restore the symmetry between the two layers.
For a symmetric augmenting path belonging to type (ii),
the path itself has a symmetric form around the link edge, and therefore augmenting along it directly preserves the symmetry between the two layers.
This is why we call augmenting paths satisfying either of the above conditions symmetric augmenting paths.

\subparagraph*{Invariants.}
We state the main invariants maintained by our algorithm.
First, the algorithm maintains the \emph{Alternating Cycle Invariant} (ACI), which ensures that no alternating cycle has negative net cost, which is crucial in analyzing the approximation guarantee and the time complexity.
In addition, the prism graph consists of two geometric layers connected by link edges, and we maintain the two layers symmetrically, which is called \emph{Layer Symmetry Invariant} (LSI).
By using symmetric augmenting paths, whenever an augmenting path changes the matching in one layer, the corresponding symmetric change is made in the other layer.
\begin{description}
\item[(ACI)] No alternating cycle with negative net cost with respect to $M$ exists.
\item[(LSI)] The matchings in the upper and lower layers are maintained symmetrically.
\end{description}
The invariant ACI is the key correctness condition used in the approximation analysis.
The proof of this invariant is given in Section~\ref{subsec:correctness_and_bound}.
The invariant LSI is maintained by using symmetric augmenting paths.
In Section~\ref{subsec:correctness_and_bound}, we prove that among all minimum-net-cost augmenting paths, there always exists a symmetric augmenting path.
We also prove that, if the selected augmenting path lies entirely in one layer, then augmenting along its mirrored path in the other layer is safe and preserves the minimum-net-cost property.
\begin{restatable}[Alternating Cycle Invariant]{lemma}{aci}\label{lem:aci}
For every phase $i \ge 0$, no alternating cycle has a negative net cost under $M_i$.
\end{restatable}
\begin{restatable}[Layer Symmetry Invariant]{lemma}{lsi}\label{lem:lsi}
For any phase $i$, if the matching $M_{i-1}$ is symmetric, a symmetric matching is obtained within at most two minimum net cost augmentation phases.
\end{restatable}

It remains to find a symmetric augmenting path efficiently.
The following lemma states that the data structure supports both finding a symmetric augmenting path and updating the matching within the desired time bound.
Its proof is given in Section~\ref{sec:Data-Structure}.

\begin{restatable}{lemma}{DataStructureLemma}
\label{lem:data-structure-time-complexity}
Assuming ACI and LSI, there exists a data structure that maintains $M$ and returns a compact symmetric augmenting path $\Pi$ in the prism graph in $|\Pi|\cdot\poly(\log \Delta,1/\varepsilon)$ time.
After augmenting along $\Pi$, the data structure can be updated within the same asymptotic time bound.
\end{restatable}

\subsection{Analysis}
We briefly show that our algorithm returns a $(1+\varepsilon)$-approximate matching with penalties in near-linear time.
We defer the detailed proofs of the following lemmas to Section~\ref{sec:correctness_and_time_complexity}.

The invariant ACI is used in the proof of the following lemma.
\begin{restatable}{lemma}{approxbound}\label{lem:approx_bound}
The expected cost of the matching $\tilde{M}_f$ we have at the end of the algorithm under $\tilde{\mathrm{w}}$ satisfies $\mathbb{E}[\tilde{\mathrm{w}}(\tilde{M}_f)] \le (1+\varepsilon) \cdot \tilde{W}^*$, where $\tilde{W}^*$ is the optimal cost in $\tilde{G}$ under $\tilde{\mathrm{w}}$.
\end{restatable}
For the analysis of the time complexity, we first show that 
the expected total length of all augmenting paths is near-linear
using the invariant ACI.
\begin{restatable}{lemma}{length}
\label{lem:length-restriction}
Let $\Pi_1,\Pi_2,\ldots,\Pi_n$ be the augmenting paths produced by the algorithm.
Then $\mathbb{E}[\sum_{i=1}^n |\Pi_i|]=O((n/\varepsilon)\log n)$.
\end{restatable}
By combining this with the preparation step that decomposes an instance into subinstances with bounded spread and Lemma~\ref{lem:data-structure-time-complexity}, we have the following theorem.
\main*

\section{Analysis of the Algorithm}\label{sec:correctness_and_time_complexity}
In this section, we analyze the correctness and time complexity of our algorithm. 
We establish the \emph{Alternating Cycle Invariant} (ACI) and the \emph{Layer Symmetry Invariant} (LSI) to prove the $(1+\varepsilon)$-approximation guarantee.
We then derive a near-linear running time by bounding the total length of the augmenting paths. 
These establish the theoretical guarantees for our approach.

\subsection{Correctness and Approximation Bound}\label{subsec:correctness_and_bound}
We prove the correctness of our algorithm and analyze its approximation bound.
We first prove the ACI and the LSI.
We then show that the resulting matching is a $(1+\varepsilon)$-approximation of the optimal matching.

We define the terminology for edges affected by an augmentation.
Recall that in phase $i$, the matching $M_i$ is computed by augmenting $M_{i-1}$ along a compact minimum net-cost augmenting path $\Pi_i$.
An edge in the prism graph $\tilde{G}$ of $G$ is \emph{affected} by $\Pi_i$ if it shares at least one vertex with $\Pi_i$.
An affected edge $e \notin \Pi_i$ is \emph{reducing} if $\Phi_{M_{i-1}}(e) > \Phi_{M_i}(e)$.
Therefore, a reducing edge is a geometric edge that is non-local under $M_{i-1}$ and local under $M_i$.

To prove the two invariants, we first establish a property of reducing edges.
Let $C_i$ be an alternating cycle with a negative net cost under $M_i$, assuming such a cycle exists.
Let $\Pi_{i+1}$ be a minimum net-cost augmenting path under $M_i$.
We assume both $C_i$ and $\Pi_{i+1}$ have the minimum number of edges among all such cycles and paths, respectively.
The following lemma characterizes how the reducing edges in these structures interact with $\Pi_i$.
\begin{lemma}\label{lem:reducing_edge_structure}
Let $e \in \Pi_i \cap M_i$ be an edge in either $C_i$ or $\Pi_{i+1}$.
If $e$ is adjacent to a reducing edge in the respective structure, then $e$ is a geometric edge, is adjacent to at most one reducing edge in the structure, and was non-local under $M_{i-1}$.
\end{lemma}
\begin{proof}
Suppose $e = (r, b)$ is adjacent to a reducing edge in the structure.
If $e$ is a link edge, its endpoints $r$ and $b$ do not belong to any geometric equivalence class under $M_i$.
Any geometric edge incident to $r$ or $b$ remains non-local under $M_i$.
This contradicts the definition of a reducing edge.
Therefore, $e$ is a geometric edge.

Suppose $e$ is adjacent to two reducing edges $(b', r)$ and $(b, r')$ in the structure.
Both reducing edges are local under $M_i$ and share endpoints with $e$.
Thus, $(b', r)$, $e$, and $(b, r')$ belong to the same equivalence class under $M_i$.
This implies $(r', b')$ also belongs to this class and shares the identical split-tree-induced distance.
The edges $(b', r)$, $e$, and $(b, r')$ form a subpath $\langle b', r, b, r' \rangle$ in the alternating structure.
If $(r', b') \in M_i$, the cycle $\langle r', b', r, b, r' \rangle$ has exactly zero net cost.
The alternating structure must traverse $(r', b')$.
If the structure is $C_i$, it cannot self-intersect, making $C_i$ exactly this zero net-cost cycle.
This contradicts the negative net cost of $C_i$.
If the structure is $\Pi_{i+1}$, it cannot contain any cycle.
If $(r', b') \notin M_i$, replacing the subpath $\langle b', r, b, r' \rangle$ with $(r', b')$ yields an alternating structure of the same type with an identical net cost but strictly fewer edges.
This contradicts the edge-minimality of $C_i$ and $\Pi_{i+1}$.
Therefore, $e$ is adjacent to at most one reducing edge.

Assume $(b', r)$ is the unique reducing edge adjacent to $e$ in the structure.
If $r$ or $b$ was unmatched in $M_{i-1}$, $e$ was non-local under $M_{i-1}$ by definition.
Otherwise, let $(r, b'')$ be the matching edge incident to $r$ under $M_{i-1}$.
Because $\Pi_i$ is a minimum net-cost augmenting path, it contains no reducing chords.
Since $(b', r)$ is a reducing edge and $r \in \Pi_i$, $b'$ must lie outside $\Pi_i$.
Thus, the matching edge incident to $b'$ remains identical under $M_{i-1}$ and $M_i$.
Let this edge be $(r'', b')$.
Because $(b', r)$ is local under $M_i$, the edges $e$ and $(r'', b')$ belong to the same equivalence class under $M_i$.
This implies their endpoints share the identical pair of refined descendants.
Assume for contradiction that $e$ was local under $M_{i-1}$.
The matching edges incident to $r$ and $b$ under $M_{i-1}$ belong to a single equivalence class.
The refined descendants defining this class contain $r$ and $b$.
Because these refined descendants lie in different child clusters, their least common ancestor is exactly the least common ancestor of $r$ and $b$.
This class under $M_{i-1}$ is uniquely determined by the least common ancestor of $r$ and $b$ along with their refined descendants.
The edge $(r, b'')$ belongs to this class.
Because $r''$ and $b'$ lie in these identical refined descendants, their least common ancestor matches that of $r$ and $b$.
Thus, the edge $(r'', b')$ also belongs to this class under $M_{i-1}$.
Consequently, the matching edges incident to $r$ and $b'$ under $M_{i-1}$ belong to the same equivalence class.
This makes $(b', r)$ local under $M_{i-1}$, contradicting its definition as a reducing edge.
Therefore, $e$ was non-local under $M_{i-1}$.
\end{proof}
We use this property of reducing edges to bound the net costs of $C_i$ and $\Pi_{i+1}$.
We apply these bounds to prove the ACI and the LSI.

\subparagraph*{Proof of the alternating cycle invariant.}
We prove the ACI by induction.
The base case trivially holds since the initial matching $M_0$ is empty.
For the inductive step, we assume that no alternating cycle with a negative net cost exists up to phase $i-1$.
Suppose for the sake of contradiction that there exists an alternating cycle $C_i$ with a negative net cost with respect to $M_i$.
Without loss of generality, we assume $C_i$ has the minimum number of edges among all such cycles.
Under these assumptions, we show that $\phi_{M_i}(C_i) \ge 0$, which contradicts the negative net cost of $C_i$ and establishes the invariant.

We first prove that $C_i$ must share an edge in $M_i$ with the augmenting path $\Pi_i$.
\begin{lemma}\label{lem:ci_intersection}
If $C_i$ is an alternating cycle with a negative net cost under $M_i$, then $C_i$ contains at least one edge in $\Pi_i \cap M_i$.
\end{lemma}
\begin{proof}
We first show that $C_i$ shares at least one vertex with $\Pi_i$.
Suppose $C_i$ shares no vertices with $\Pi_i$.
The matching edges incident to the vertices of $C_i$ remain identical under $M_{i-1}$ and $M_i$.
This preserves the local or non-local classification of every geometric edge in $C_i$.
The adjusted cost of every link edge is fixed.
Consequently, the adjusted cost of every edge in $C_i$ is identical under $M_{i-1}$ and $M_i$.
This yields $\phi_{M_{i-1}}(C_i) = \phi_{M_i}(C_i) < 0$.
This contradicts the induction hypothesis that no negative net cost alternating cycle exists under $M_{i-1}$.
Thus, $C_i$ shares at least one vertex with $\Pi_i$.

We now show that $C_i$ contains an edge in $\Pi_i \cap M_i$.
Since $C_i$ is a cycle and $\Pi_i$ is a path, $C_i$ cannot be a subgraph of $\Pi_i$.
As $C_i$ shares a vertex with $\Pi_i$, it contains an edge $(u, v) \notin \Pi_i$ such that $u \in \Pi_i$.
Because $u \in \Pi_i$, the edge in $M_i$ incident to $u$ lies entirely on $\Pi_i$.
This implies $(u, v) \notin M_i$.
Since $C_i$ is an alternating cycle under $M_i$ containing the unmatched edge $(u, v)$, it must traverse the edge in $M_i$ incident to $u$.
This edge belongs to $\Pi_i \cap M_i$.
\end{proof}
This lemma establishes that any negative net cost cycle under $M_i$ intersects $\Pi_i$ at one or more matched edges.

\begin{figure}
    \centering
    \includegraphics[width=0.9\linewidth]{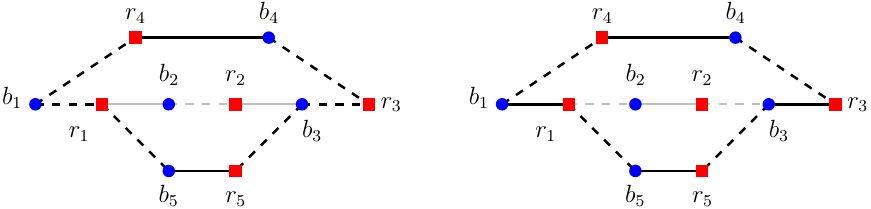}
    \caption{Illustration of $C_i$ and its decomposition. The left and right figures show the matching statuses under $M_{i-1}$ and $M_i$, respectively. Solid lines denote edges in the matching, and dashed lines denote unmatched edges. Vertices in $B_0 \cup B_1$ and $R_0 \cup R_1$ are represented by circles and squares, respectively. The augmenting path is $\Pi_i = \langle b_1, r_1, b_2, r_2, b_3, r_3 \rangle$. Under $M_i$, the alternating cycle is $C_i = \langle b_1, r_1, b_5, r_5, b_3, r_3, b_4, r_4, b_1 \rangle$. The decomposition yields $\mathcal{C}^0 = \{ \langle b_1, r_1 \rangle, \langle b_3, r_3 \rangle \}$, $\mathcal{C}^+ = \{ \langle b_1, r_4, b_4, r_3 \rangle \}$, and $\mathcal{C}^- = \{ \langle r_1, b_5, r_5, b_3    \rangle \}$.}
    \label{fig:Ci_decomposition_eg}
\end{figure}

By Lemma~\ref{lem:ci_intersection}, we decompose $C_i$ into subpaths to evaluate its net cost.
We partition $C_i$ into three sets of alternating paths under $M_i$: $\mathcal{C}^0$, $\mathcal{C}^+$, and $\mathcal{C}^-$.
Figure~\ref{fig:Ci_decomposition_eg} illustrates an example of the decomposition.
Let $\mathcal{C}^0$ be the set of connected components of $C_i \cap \Pi_i$, where each component is called a \emph{shared} component.
The proof of Lemma~\ref{lem:ci_intersection} shows that no edge in $C_i \setminus \Pi_i$ sharing a vertex with $\Pi_i$ belongs to $M_i$.
Thus, the first and last edges of every component in $\mathcal{C}^0$ belong to $M_i$.
Let $\mathcal{C}^+$ and $\mathcal{C}^-$ partition the connected components of $C_i \setminus \Pi_i$.
Every component $\pi \in \mathcal{C}^+ \cup \mathcal{C}^-$ begins and ends with an edge outside $M_i$.
Since $\pi$ is an alternating path with edges not in a matching at both ends, its two endpoints on $\Pi_i$ lie in opposite partitions $B_0 \cup B_1$ and $R_0 \cup R_1$.
Let $u \in B_0 \cup B_1$ and $v \in R_0 \cup R_1$ be these endpoints.
Since $\Pi_i$ is directed from $B_0 \cup B_1$ to $R_0 \cup R_1$, we classify $\pi$ as a \emph{forward} component in $\mathcal{C}^+$ if $u$ precedes $v$ along $\Pi_i$, and as a \emph{backward} component in $\mathcal{C}^-$ if $v$ precedes $u$.
Using this decomposition, we express the net cost of $C_i$ as
\begin{equation}\label{eq:net_cost_ci}
    \phi_{M_i}(C_i) = \sum_{\pi \in \mathcal{C}^+} \phi_{M_i}(\pi) + \sum_{\pi \in \mathcal{C}^-} \phi_{M_i}(\pi) + \sum_{\pi \in \mathcal{C}^0} \phi_{M_i}(\pi).
\end{equation}

We establish lower bounds on the net cost for each partition of $C_i$.
Let $n^0$ denote the number of geometric edges in $\mathcal{C}^0 \cap M_i$ that were non-local under $M_{i-1}$.
Let $r^+$ and $r^-$ denote the number of reducing edges in $\mathcal{C}^+$ and $\mathcal{C}^-$, respectively.
For any component $\pi \in \mathcal{C}^+ \cup \mathcal{C}^-$, let $\overline{\pi}$ denote the subpath of $\Pi_i$ between the endpoints of $\pi$.
The following lemma establishes the lower bounds.
\begin{lemma}\label{lem:ci_bounds}
Suppose $C_i$ is an alternating cycle under $M_i$ with a negative net cost and the minimum number of edges.
The net costs of the forward components $\mathcal{C}^+$, the backward components $\mathcal{C}^-$, and the shared components $\mathcal{C}^0$ satisfy 
\begin{align}
    \sum_{\pi \in \mathcal{C}^+} \phi_{M_i}(\pi) &\ge \sum_{\pi \in \mathcal{C}^+} \phi_{M_{i-1}}(\overline{\pi}) - r^+\theta, \label{eq:bound_c_plus} \\
    \sum_{\pi \in \mathcal{C}^-} \phi_{M_i}(\pi) &\ge -\sum_{\pi \in \mathcal{C}^-} \phi_{M_{i-1}}(\overline{\pi}) - r^-\theta, \label{eq:bound_c_minus} \\
    \sum_{\pi \in \mathcal{C}^0} \phi_{M_i}(\pi) &\ge -\sum_{\pi \in \mathcal{C}^0} \phi_{M_{i-1}}(\pi) + n^0\theta. \label{eq:bound_c_zero}
\end{align}
\end{lemma}
\begin{proof}
For each forward component $\pi \in \mathcal{C}^+$, replacing $\overline{\pi}$ with $\pi$ along $\Pi_i$ yields an augmenting path under $M_{i-1}$. 
Since $\Pi_i$ is a minimum net-cost augmenting path under $M_{i-1}$, we obtain $\phi_{M_{i-1}}(\pi) \ge \phi_{M_{i-1}}(\overline{\pi})$.
The adjusted costs of the link edges remain identical under $M_{i-1}$ and $M_i$.
The adjusted costs of the geometric edges do not decrease, except for the reducing edges which decrease by $\theta$. 
Subtracting $\theta$ for each of the $r^+$ reducing edges across all forward components yields $\sum_{\pi \in \mathcal{C}^+} \phi_{M_i}(\pi) \ge \sum_{\pi \in \mathcal{C}^+} \phi_{M_{i-1}}(\pi) - r^+\theta \ge \sum_{\pi \in \mathcal{C}^+} \phi_{M_{i-1}}(\overline{\pi}) - r^+\theta$.
This establishes~\eqref{eq:bound_c_plus}.

For each backward component $\pi \in \mathcal{C}^-$, the union $\pi \cup \overline{\pi}$ forms an alternating cycle under $M_{i-1}$. 
By the induction hypothesis, its net cost is non-negative, which yields $\phi_{M_{i-1}}(\pi) + \phi_{M_{i-1}}(\overline{\pi}) \ge 0$. 
Subtracting $\theta$ for each of the $r^-$ reducing edges across all backward components yields $\sum_{\pi \in \mathcal{C}^-} \phi_{M_i}(\pi) \ge \sum_{\pi \in \mathcal{C}^-} \phi_{M_{i-1}}(\pi) - r^-\theta \ge -\sum_{\pi \in \mathcal{C}^-} \phi_{M_{i-1}}(\overline{\pi}) - r^-\theta$.
This proves~\eqref{eq:bound_c_minus}.

For every shared component $\pi \in \mathcal{C}^0$, its edges lie entirely on $\Pi_i$.
This implies $\pi \setminus M_{i-1} = \pi \cap M_i$ and $\pi \cap M_{i-1} = \pi \setminus M_i$.
The net cost of $\pi$ under $M_{i-1}$ is $\phi_{M_{i-1}}(\pi) = \Phi_{M_{i-1}}(\pi \setminus M_{i-1}) - \Phi_{M_{i-1}}(\pi \cap M_{i-1})$.
Substituting the edge relations yields $\phi_{M_{i-1}}(\pi) = \Phi_{M_{i-1}}(\pi \cap M_i) - \Phi_{M_{i-1}}(\pi \setminus M_i)$.
The geometric edges in $\pi \setminus M_i$ were matched and local under $M_{i-1}$, meaning their adjusted costs do not decrease under $M_i$.
The geometric edges in $\pi \cap M_i$ are matched and local under $M_i$.
Exactly $n^0$ geometric edges across all shared components were non-local under $M_{i-1}$.
This yields $\sum_{\pi \in \mathcal{C}^0} \Phi_{M_{i-1}}(\pi \cap M_i) = \sum_{\pi \in \mathcal{C}^0} \Phi_{M_i}(\pi \cap M_i) + n^0\theta$.
Applying these bounds gives $-\sum_{\pi \in \mathcal{C}^0} \phi_{M_{i-1}}(\pi) \le \sum_{\pi \in \mathcal{C}^0} \Phi_{M_i}(\pi \setminus M_i) - \sum_{\pi \in \mathcal{C}^0} \Phi_{M_i}(\pi \cap M_i) - n^0\theta$.
Since $\Phi_{M_i}(\pi \setminus M_i) - \Phi_{M_i}(\pi \cap M_i) = \phi_{M_i}(\pi)$, rearranging these terms establishes~\eqref{eq:bound_c_zero}.
\end{proof}
This lemma bounds the net cost of each partition in $C_i$.

Summing~\eqref{eq:bound_c_plus},~\eqref{eq:bound_c_minus}, and~\eqref{eq:bound_c_zero} yields a lower bound on the net cost of $C_i$.
Since $C_i$ is a cycle, every edge on $\Pi_i$ is covered by $\overline{\pi}$ for $\pi \in \mathcal{C}^+$ exactly as many times as it is covered by $\overline{\pi}$ for $\pi \in \mathcal{C}^-$ and by $\pi$ for $\pi \in \mathcal{C}^0$ combined.
This implies $\sum_{\pi \in \mathcal{C}^+} \phi_{M_{i-1}}(\overline{\pi}) - \sum_{\pi \in \mathcal{C}^-} \phi_{M_{i-1}}(\overline{\pi}) - \sum_{\pi \in \mathcal{C}^0} \phi_{M_{i-1}}(\pi) = 0$.
Applying this equality to the sum yields
\begin{equation}\label{eq:bound_ci}
    \phi_{M_i}(C_i) \ge (n^0 - r^+ - r^-)\theta.
\end{equation}
This inequality establishes that $C_i$ can have a negative net cost only if the total number of reducing edges in $\mathcal{C}^+$ and $\mathcal{C}^-$ strictly exceeds the number of geometric edges in $\mathcal{C}^0 \cap M_i$ that were non-local under $M_{i-1}$.

We now formally establish the ACI.

\aci*
\begin{proof}
We proceed by induction on $i$.
The base case holds because the initial matching $M_0$ is empty.
Assume the invariant holds up to phase $i-1$.
Suppose for contradiction that an alternating cycle with a negative net cost exists under $M_i$.
Let $C_i$ be such a cycle with the minimum number of edges.

By definition, a reducing edge lies outside $\Pi_i$ but shares at least one vertex with $\Pi_i$.
Since every vertex of $\Pi_i$ is matched within $\Pi_i$ under $M_i$, the reducing edge cannot belong to $M_i$.
As $C_i$ is an alternating cycle under $M_i$, it must traverse the edge in $M_i$ incident to the shared vertex.
Because this shared vertex lies on $\Pi_i$, this edge lies entirely on $\Pi_i$.
Thus, this edge in the matching belongs to $C_i \cap \Pi_i \cap M_i$, which is $\mathcal{C}^0 \cap M_i$.
Therefore, every reducing edge in $C_i$ is adjacent to at least one edge in $\mathcal{C}^0 \cap M_i$.

By Lemma~\ref{lem:reducing_edge_structure}, each edge in $\mathcal{C}^0 \cap M_i$ is adjacent to at most one reducing edge in $C_i$, and was non-local under $M_{i-1}$.
Therefore, the $r^+ + r^-$ reducing edges in $C_i$ map injectively to distinct geometric edges in $\mathcal{C}^0 \cap M_i$ that were non-local under $M_{i-1}$.
Since $n^0$ denotes the number of such geometric edges, this establishes $n^0 \ge r^+ + r^-$.
Applying this bound to~\eqref{eq:bound_ci} yields $\phi_{M_i}(C_i) \ge 0$.
This contradicts the assumption that $C_i$ has a negative net cost.
\end{proof}

\subparagraph*{Proof of the layer symmetry invariant.}
We prove the LSI by induction.
The base case holds since the initial matching $M_0$ is empty.
For the inductive step, we assume the matching $M_{i-1}$ is symmetric.
Let $\Pi_i$ be a minimum net-cost augmenting path computed under $M_{i-1}$ starting from a free vertex in $B_0$.
We analyze two cases based on whether $\Pi_i$ contains a link edge.
In either case, we show that the algorithm reaches a symmetric matching within at most two minimum net cost augmentation phases.

\begin{figure}
    \centering
    \includegraphics[width=0.75\linewidth]{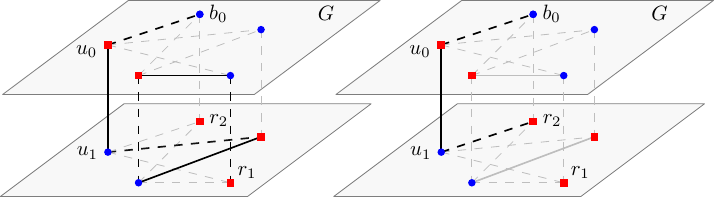}
    \caption{Illustration of an augmenting path $\Pi_i$ and its corresponding symmetric augmenting path $\Pi_i^1$. The left and right figures show the paths $\Pi_i$ and $\Pi_i^1$ highlighted in black, respectively. Solid lines denote matched edges, and dashed lines denote unmatched edges. Vertices in $B_0 \cup B_1$ and $R_0 \cup R_1$ are represented by circles and squares, respectively. In the left figure, $\Pi_i$ decomposes into the upper-layer geometric subpath $\pi_1 = \langle b_0, u_0 \rangle$, the first link edge $e = (u_0, u_1)$, and the remaining subpath $\pi_2$ from $u_1$ to $r_1$. In the right figure, the mirrored path of $\pi_1$ is $\hat{\pi}_1 = \langle u_1, r_2 \rangle$. The symmetric augmenting path $\Pi_i^1$ consists of $\pi_1$, the link edge $e$, and $\hat{\pi}_1$, which yields $\Pi_i^1 = \langle b_0, u_0, u_1, r_2 \rangle$.}
    \label{fig:aug_with_link}
\end{figure}

For any alternating path $\pi$, let $\hat{\pi}$ denote its \emph{mirrored path}.
We obtain $\hat{\pi}$ by replacing every geometric edge in $\pi$ with the geometric edge connecting the corresponding endpoint copies in the opposite layer, while traversing the identical link edges.
An augmenting path $\Pi$ is a \emph{symmetric augmenting path} if it satisfies one of two conditions.
First, $\Pi$ consists entirely of geometric edges in a single layer.
Second, $\Pi$ contains exactly one link edge $e$ and satisfies $\Pi = \pi \circ e \circ \hat{\pi}$ for some single-layer geometric subpath $\pi$.

We first consider the case where $\Pi_i$ contains at least one link edge.
Let $b_0 \in B_0$ be the starting free vertex of $\Pi_i$.
Let $e = (u_0, u_1)$ be the first link edge along $\Pi_i$, with $u_0$ in the upper layer and $u_1$ in the lower layer.
We decompose $\Pi_i$ into the subpath $\pi_1$ from $b_0$ to $u_0$, the edge $e$, and the remaining subpath $\pi_2$ from $u_1$ to the end of $\Pi_i$.
Since $\pi_1$ precedes the first link edge, it consists entirely of geometric edges in the upper layer.
Let $\Pi_i^1$ be the path formed by $\pi_1$, the edge $e$, and $\hat{\pi}_1$.
Figure~\ref{fig:aug_with_link} illustrates $\Pi_i$ and the construction of $\Pi_i^1$.
The following lemma establishes that $\Pi_i^1$ is a symmetric minimum net cost augmenting path under $M_{i-1}$.

\begin{lemma}\label{lem:lsi_link}
Suppose $M_{i-1}$ is symmetric and $\Pi_i$ contains at least one link edge.
Then $\Pi_i^1$ is a symmetric minimum net cost augmenting path under $M_{i-1}$.
\end{lemma}
\begin{proof}
Since $M_{i-1}$ is symmetric, replacing geometric edges with their copies in the opposite layer preserves their matching statuses and adjusted costs.
Thus, $\hat{\pi}_1$ and $\hat{\pi}_2$ form valid alternating paths under $M_{i-1}$.
Their net costs satisfy $\phi_{M_{i-1}}(\hat{\pi}_1) = \phi_{M_{i-1}}(\pi_1)$ and $\phi_{M_{i-1}}(\hat{\pi}_2) = \phi_{M_{i-1}}(\pi_2)$.

Since $\pi_1$ lies entirely in the upper layer, $\hat{\pi}_1$ lies entirely in the lower layer.
Thus, $\pi_1$ and $\hat{\pi}_1$ share no vertices.
This implies $\Pi_i^1$ forms a simple alternating path connecting two free vertices.
As $\Pi_i^1$ is a valid augmenting path, we obtain $\phi_{M_{i-1}}(\Pi_i^1) \ge \phi_{M_{i-1}}(\Pi_i)$.

Let $\Pi_i^2$ be the alternating walk formed by $\hat{\pi}_2$, the edge $e$, and $\pi_2$.
This walk connects two free vertices and decomposes into an augmenting path and a set of alternating cycles.
By Lemma~\ref{lem:aci}, no alternating cycle has a negative net cost under $M_{i-1}$.
Thus, the net cost of $\Pi_i^2$ is bounded below by the net cost of the augmenting path, which yields $\phi_{M_{i-1}}(\Pi_i^2) \ge \phi_{M_{i-1}}(\Pi_i)$.

The sum of their net costs is
\begin{align}
    \phi_{M_{i-1}}(\Pi_i^1) + \phi_{M_{i-1}}(\Pi_i^2) &= \left(2\phi_{M_{i-1}}(\pi_1) + \phi_{M_{i-1}}(e)\right) + \left(2\phi_{M_{i-1}}(\pi_2) + \phi_{M_{i-1}}(e)\right) \nonumber \\
    &= 2\left(\phi_{M_{i-1}}(\pi_1) + \phi_{M_{i-1}}(e) + \phi_{M_{i-1}}(\pi_2)\right) \nonumber \\
    &= 2\phi_{M_{i-1}}(\Pi_i). \label{eq:sum_net_costs_lsi}
\end{align}
We previously established the lower bounds $\phi_{M_{i-1}}(\Pi_i^1) \ge \phi_{M_{i-1}}(\Pi_i)$ and $\phi_{M_{i-1}}(\Pi_i^2) \ge \phi_{M_{i-1}}(\Pi_i)$.
The equation~\eqref{eq:sum_net_costs_lsi} shows their sum is exactly $2\phi_{M_{i-1}}(\Pi_i)$.
This forces $\phi_{M_{i-1}}(\Pi_i^1) = \phi_{M_{i-1}}(\Pi_i)$, proving that $\Pi_i^1$ is a minimum net cost augmenting path.
By construction, $\Pi_i^1$ is a symmetric augmenting path.
\end{proof}
Since $\Pi_i^1$ is a symmetric minimum net cost augmenting path, augmenting $M_{i-1}$ along $\Pi_i^1$ changes the geometric edges identically in both layers.
This ensures the resulting matching $M_i$ remains symmetric.

We next consider the case where $\Pi_i$ contains no link edge.
Since $\Pi_i$ starts in $B_0$ and contains no link edge, it consists entirely of geometric edges in the upper layer.
Let $\hat{\Pi}_i$ be the mirrored path of $\Pi_i$.
As $M_{i-1}$ is symmetric, $\hat{\Pi}_i$ forms an augmenting path under $M_{i-1}$ with net cost $\phi_{M_{i-1}}(\hat{\Pi}_i) = \phi_{M_{i-1}}(\Pi_i)$.
Because augmenting $M_{i-1}$ along $\Pi_i$ changes edges only in the upper layer, $\hat{\Pi}_i$ remains a valid augmenting path under $M_i$.
Its net cost remains identical, i.e., $\phi_{M_i}(\hat{\Pi}_i) = \phi_{M_{i-1}}(\Pi_i)$.
The following lemma establishes that $\hat{\Pi}_i$ is a symmetric minimum net-cost augmenting path for phase $i+1$.

\begin{lemma}\label{lem:lsi_no_link}
Suppose $M_{i-1}$ is symmetric and $\Pi_i$ contains no link edge.
Then $\hat{\Pi}_i$ is a symmetric minimum net cost augmenting path with respect to $M_i$.
\end{lemma}
\begin{proof}
Let $\Pi_{i+1}$ be an arbitrary minimum net-cost augmenting path under $M_i$.
Assume $\Pi_{i+1}$ has the minimum number of edges among all such paths.
Let $\Gamma = \Pi_i \oplus \Pi_{i+1}$.
The symmetric difference $\Gamma$ can be decomposed into two augmenting paths $A_1$ and $A_2$ with respect to $M_{i-1}$, and a possibly empty set $\mathcal{C}$ of alternating cycles under $M_{i-1}$.
Since $\Pi_i$ is a minimum net cost augmenting path under $M_{i-1}$, both $\phi_{M_{i-1}}(A_1)$ and $\phi_{M_{i-1}}(A_2)$ are bounded below by $\phi_{M_{i-1}}(\Pi_i)$.
By Lemma~\ref{lem:aci}, $\phi_{M_{i-1}}(C) \ge 0$ holds for every cycle $C \in \mathcal{C}$.
Thus, the net cost of $\Gamma$ under $M_{i-1}$ satisfies $\phi_{M_{i-1}}(\Gamma) \ge 2\phi_{M_{i-1}}(\Pi_i)$.

We evaluate the difference $\phi_{M_{i-1}}(\Pi_i) + \phi_{M_i}(\Pi_{i+1}) - \phi_{M_{i-1}}(\Gamma)$ by partitioning the edges of $\Pi_i \cup \Pi_{i+1}$:
\begin{equation}\label{eq:cost_diff}
    \sum_{e \in \Pi_{i+1} \setminus \Pi_i} \left( \phi_{M_i}(e) - \phi_{M_{i-1}}(e) \right) + \sum_{e \in \Pi_i \cap \Pi_{i+1}} \left( \phi_{M_i}(e) + \phi_{M_{i-1}}(e) \right).
\end{equation}
Since $M_i = M_{i-1} \oplus \Pi_i$, the edges in $\Pi_{i+1} \setminus \Pi_i$ maintain identical matching statuses under $M_{i-1}$ and $M_i$.
Their adjusted costs decrease by $\theta$ under $M_i$ if and only if they are reducing edges.
Let $r_1$ be the number of reducing edges in $\Pi_{i+1} \setminus \Pi_i$.
This bounds the first sum in~\eqref{eq:cost_diff} below by $-r_1\theta$.

The edges in $\Pi_i \cap \Pi_{i+1}$ have opposite matching statuses under $M_{i-1}$ and $M_i$.
For any edge $e \in (\Pi_i \cap \Pi_{i+1}) \cap M_{i-1}$, it is unmatched under $M_i$.
This yields $\phi_{M_i}(e) + \phi_{M_{i-1}}(e) = \Phi_{M_i}(e) - \Phi_{M_{i-1}}(e) \ge 0$.
For any edge $e \in (\Pi_i \cap \Pi_{i+1}) \setminus M_{i-1}$, it is matched under $M_i$.
This yields $\phi_{M_i}(e) + \phi_{M_{i-1}}(e) = \Phi_{M_{i-1}}(e) - \Phi_{M_i}(e)$.
Since $e$ becomes matched and local under $M_i$, this difference exactly equals $\theta$ if $e$ was non-local under $M_{i-1}$, and is non-negative otherwise.
Let $r_2$ be the number of such non-local geometric edges in $(\Pi_i \cap \Pi_{i+1}) \setminus M_{i-1}$.
This bounds the second sum in~\eqref{eq:cost_diff} below by $r_2\theta$.
Combining these bounds yields $\phi_{M_{i-1}}(\Pi_i) + \phi_{M_i}(\Pi_{i+1}) - \phi_{M_{i-1}}(\Gamma) \ge (r_2 - r_1)\theta$.

Since adjusted costs decrease exclusively for geometric edges sharing a vertex with $\Pi_i$, every reducing edge in $\Pi_{i+1} \setminus \Pi_i$ is adjacent to at least one vertex of $\Pi_i$.
As $\Pi_{i+1}$ is an alternating path under $M_i$, the edge in $\Pi_{i+1} \cap M_i$ incident to this shared vertex lies entirely on $\Pi_i$.
This edge belongs to $\Pi_{i+1} \cap \Pi_i \cap M_i$, which exactly equals $(\Pi_i \cap \Pi_{i+1}) \setminus M_{i-1}$.
By Lemma~\ref{lem:reducing_edge_structure}, each edge in $(\Pi_i \cap \Pi_{i+1}) \setminus M_{i-1}$ is adjacent to at most one reducing edge in $\Pi_{i+1} \setminus \Pi_i$, and was non-local under $M_{i-1}$.
This injective mapping establishes $r_2 \ge r_1$.
Applying this to the cost difference yields $\phi_{M_{i-1}}(\Pi_i) + \phi_{M_i}(\Pi_{i+1}) - \phi_{M_{i-1}}(\Gamma) \ge 0$.

Rearranging this inequality and applying the lower bound on $\phi_{M_{i-1}}(\Gamma)$ yields $\phi_{M_{i-1}}(\Pi_i) + \phi_{M_i}(\Pi_{i+1}) \ge \phi_{M_{i-1}}(\Gamma) \ge 2\phi_{M_{i-1}}(\Pi_i)$.
This simplifies to $\phi_{M_i}(\Pi_{i+1}) \ge \phi_{M_{i-1}}(\Pi_i)$.
This establishes that any augmenting path under $M_i$ has a net cost of at least $\phi_{M_{i-1}}(\Pi_i)$.

We previously established that $\hat{\Pi}_i$ is an augmenting path under $M_i$ with net cost exactly $\phi_{M_{i-1}}(\Pi_i)$.
Thus, $\hat{\Pi}_i$ is a minimum net cost augmenting path under $M_i$.
Since $\hat{\Pi}_i$ consists entirely of geometric edges confined to the lower layer, it perfectly satisfies the definition of a symmetric augmenting path.
\end{proof}
By Lemma~\ref{lem:lsi_no_link}, we can select the symmetric augmenting path $\hat{\Pi}_i$ as the minimum net cost augmenting path for phase $i+1$.
Augmenting $M_i$ along $\hat{\Pi}_i$ changes the lower layer symmetrically to the upper layer.
This ensures the resulting matching $M_{i+1}$ is symmetric.

We now formally establish the LSI.

\lsi*
\begin{proof}
We proceed by induction on $i$.
The base case holds since the initial matching $M_0$ is empty.
For the inductive step, we assume the matching $M_{i-1}$ is symmetric.
Let $\Pi_i$ be a minimum net cost augmenting path computed under $M_{i-1}$.
If $\Pi_i$ contains at least one link edge, Lemma~\ref{lem:lsi_link} establishes that we can transform $\Pi_i$ into a symmetric minimum net cost augmenting path.
Augmenting along this symmetric path yields a symmetric matching $M_i$.
If $\Pi_i$ contains no link edge, it is confined to a single layer.
Lemma~\ref{lem:lsi_no_link} guarantees that the mirrored path of $\Pi_i$ is a minimum net cost augmenting path with respect to $M_i$.
By augmenting along this mirrored path in phase $i+1$, the algorithm yields a symmetric matching $M_{i+1}$.
In either case, the algorithm maintains a symmetric matching within at most two augmentation phases.
\end{proof}

We analyze the approximation bound of the algorithm.
Let $\tilde{M}_f$ be the matching in $\tilde{G}$ returned by the algorithm.
By Lemma~\ref{lem:lsi}, $\tilde{M}_f$ is symmetric.
We define two cost functions for a perfect matching $\tilde{M}$ in $\tilde{G}$.
Let $\tilde{\mathrm{w}}_{\sdist}(\tilde{M})$ denote the cost of $\tilde{M}$ under the split-tree-induced distance, defined as the sum of $\sdist(e)$ for its geometric edges and $\tilde{c}(e)$ for its link edges.
Let $\tilde{\mathrm{w}}(\tilde{M})$ denote the cost of $\tilde{M}$ under the original metric $d$, defined as the sum of $d(e)$ for its geometric edges and $\tilde{c}(e)$ for its link edges.
The following lemma establishes the approximation bound of the algorithm.

\approxbound*
\begin{proof}
Let $\tilde{M}^*$ be an optimal matching in $\tilde{G}$ such that $\tilde{\mathrm{w}}(\tilde{M}^*) = \tilde{W}^*$.
The symmetric difference $\tilde{M}_f \oplus \tilde{M}^*$ decomposes into a set $\mathcal{C}$ of vertex-disjoint alternating cycles.
By Lemma~\ref{lem:aci}, no alternating cycle in $\mathcal{C}$ has a negative net cost under $\tilde{M}_f$.
Summing the net costs of all cycles in $\mathcal{C}$ yields $\sum_{C \in \mathcal{C}} \phi_{\tilde{M}_f}(C) \ge 0$.
Because $\tilde{M}_f$ and $\tilde{M}^*$ are perfect matchings, expanding the net costs yields $\Phi_{\tilde{M}_f}(\tilde{M}^*) - \Phi_{\tilde{M}_f}(\tilde{M}_f) \ge 0$.

Every edge in $\tilde{M}_f$ is local under $\tilde{M}_f$.
Thus, its adjusted cost exactly equals its split-tree-induced cost.
This yields $\Phi_{\tilde{M}_f}(\tilde{M}_f) = \tilde{\mathrm{w}}_{\sdist}(\tilde{M}_f)$.
For the optimal matching $\tilde{M}^*$, the adjusted cost of each link edge equals its original cost.
The adjusted cost of each geometric edge in $\tilde{M}^*$ is bounded above by its split-tree-induced distance plus $\theta$.
Since a perfect matching in $\tilde{G}$ contains exactly $n$ edges, $\tilde{M}^*$ contains at most $n$ geometric edges.
This bounds the adjusted cost of $\tilde{M}^*$ by $\Phi_{\tilde{M}_f}(\tilde{M}^*) \le \tilde{\mathrm{w}}_{\sdist}(\tilde{M}^*) + n\theta$.
Substituting these values into the net cost inequality yields $\tilde{\mathrm{w}}_{\sdist}(\tilde{M}_f) \le \tilde{\mathrm{w}}_{\sdist}(\tilde{M}^*) + n\theta$.

By Lemma~\ref{lem:split-tree-distance-approx}, we obtain $\tilde{\mathrm{w}}(\tilde{M}_f) \le \tilde{\mathrm{w}}_{\sdist}(\tilde{M}_f)$ and $\mathbb{E}[\tilde{\mathrm{w}}_{\sdist}(\tilde{M}^*)] \le (1+\varepsilon/2)\tilde{W}^*$.
Taking the expectation on both sides yields $\mathbb{E}[\tilde{\mathrm{w}}(\tilde{M}_f)] \le (1+\varepsilon/2)\tilde{W}^* + n\theta$.
Since $\theta \le \frac{\varepsilon \tilde{W}^*}{3n}$, this implies $\mathbb{E}[\tilde{\mathrm{w}}(\tilde{M}_f)] \le (1 + 5\varepsilon/6)\tilde{W}^* \le (1+\varepsilon)\tilde{W}^*$.
\end{proof}
By Lemma~\ref{lem:reduction-pm-mcpm}, we can construct a bipartite matching with penalties $M_f$ for the original instance from $\tilde{M}_f$ in $O(n)$ time.
Since $\tilde{\mathrm{w}}(\tilde{M}_f) = 2\mathrm{w}(M_f)$ and $\tilde{W}^* = 2\mathrm{w}(M^*)$, Lemma~\ref{lem:approx_bound} implies $\mathbb{E}[\mathrm{w}(M_f)] \le (1+\varepsilon)\mathrm{w}(M^*)$.
This guarantees an expected $(1+\varepsilon)$-approximate optimal matching for the original problem.

\subsection{Time complexity}
We first bound the total length of the augmenting paths produced by the algorithm.
This bound is the key ingredient in the output-sensitive running-time analysis, since the data structure spends time proportional to the length of the extracted path.

We prove Lemma~\ref{lem:length-restriction} as follows.
\length*
\begin{proof}
Let $M_i=M_{i-1}\oplus \Pi_i$, with $M_0=\emptyset$.
Also, let $\tilde{M}^*$ be an optimal perfect matching in the prism graph under $\tilde{\mathrm{w}}$, with cost $\tilde{W}^*=\tilde{\mathrm{w}}(\tilde{M}^*)$.
For each path $\Pi_i$, let $\nu_i$ be the number of non-local geometric edges in $\Pi_i$.
By the definition of adjusted cost,
\begin{align}
\phi_{M_{i-1}}(\Pi_i)
=
\tilde{\mathrm{w}}_{\sdist}(M_i)-\tilde{\mathrm{w}}_{\sdist}(M_{i-1})+\theta\nu_i.
\label{eq:net-cost-weight-difference}
\end{align}
We first bound $\mathbb{E}[\sum_i \nu_i]$.
Fix an iteration $i$.
The symmetric difference $\tilde{M}^*\oplus M_{i-1}$ decomposes into alternating cycles and $n-i+1$ augmenting paths with respect to $M_{i-1}$.
Let these augmenting paths be $P_1,\ldots,P_{n-i+1}$.
For each $P_j$, we have $\phi_{M_{i-1}}(P_j)\le \Phi_{M_{i-1}}(P_j\setminus M_{i-1})$.
The edges in $\bigcup_j(P_j\setminus M_{i-1})$ are contained in $\tilde{M}^*$.
Therefore,
\begin{align}
\sum_{j=1}^{n-i+1}\phi_{M_{i-1}}(P_j)
\le
\tilde{\mathrm{w}}_{\sdist}(\tilde{M}^*)+n\theta.
\label{eq:path-net-upper-random}
\end{align}
Since $\Pi_i$ is chosen as a minimum-net-cost augmenting path, $\phi_{M_{i-1}}(\Pi_i)\le(\tilde{\mathrm{w}}_{\sdist}(\tilde{M}^*)+n\theta)/(n-i+1)$.
Summing this inequality over all iterations gives
\begin{align}
\sum_{i=1}^n\phi_{M_{i-1}}(\Pi_i)
\le
\left(\tilde{\mathrm{w}}_{\sdist}(\tilde{M}^*)+n\theta\right)H_n,
\label{eq:sum-net-cost-upper-random}
\end{align}
where $H_n$ is the $n$th harmonic number.
By Lemma~\ref{lem:split-tree-distance-approx}, we have $\mathbb{E}[\tilde{\mathrm{w}}_{\sdist}(\tilde{M}^*)]\le(1+\varepsilon/2)\tilde{W}^*$.
Taking expectation in~\eqref{eq:sum-net-cost-upper-random} and using $\theta\le \varepsilon \tilde{W}^*/(3n)$ gives $\mathbb{E}[\sum_{i=1}^n\phi_{M_{i-1}}(\Pi_i)]=O(\tilde{W}^*\log n)$.
On the other hand, summing~\eqref{eq:net-cost-weight-difference} over all iterations gives
\begin{align*}
\sum_{i=1}^n\phi_{M_{i-1}}(\Pi_i)
=
\tilde{\mathrm{w}}_{\sdist}(M_n)-\tilde{\mathrm{w}}_{\sdist}(M_0)+\theta\sum_{i=1}^n\nu_i.
\end{align*}
Since $M_n$ is a perfect matching and $\sdist$ dominates the original metric on geometric edges, we have $\tilde{\mathrm{w}}_{\sdist}(M_n)\ge \tilde{W}^*$.
Also $\tilde{\mathrm{w}}_{\sdist}(M_0)=0$.
Thus $\sum_{i=1}^n\phi_{M_{i-1}}(\Pi_i)\ge \tilde{W}^*+\theta\sum_{i=1}^n\nu_i$.
Taking expectation and combining this lower bound with the previous upper bound yields $\theta\cdot\mathbb{E}[\sum_{i=1}^n\nu_i]=O(\tilde{W}^*\log n)$.
Since $\theta\ge \varepsilon \tilde{W}^*/(6n)$, we obtain $\mathbb{E}[\sum_{i=1}^n\nu_i]=O((n/\varepsilon)\log n)$.
It remains to relate $\nu_i$ to $|\Pi_i|$.
By compactness, the number of local geometric edges in $\Pi_i$ is at most three times the number of non-local geometric edges, so the number of geometric edges in $\Pi_i$ is $O(\nu_i)$.
By the construction of the one-layer search and the LSI, each extracted prism-graph augmenting path contains only $O(1)$ link edges.
Thus $|\Pi_i|=O(\nu_i+1)$.
Therefore, $\mathbb{E}[\sum_{i=1}^n|\Pi_i|]=O(\mathbb{E}[\sum_{i=1}^n\nu_i]+n)=O((n/\varepsilon)\log n)$.
\end{proof}
Combining Lemma~\ref{lem:length-restriction} with Lemma~\ref{lem:data-structure-time-complexity}, the expected total time spent on all augmentations and data-structure updates is
\begin{align*}
\mathbb{E}\left[\sum_i |\Pi_i|\cdot\poly(\log\Delta,1/\varepsilon)\right]
=
O\left(n\log n \cdot\poly(\log\Delta,1/\varepsilon)\right).
\end{align*}
The split-tree construction takes $O(n\log\Delta)$ time and the initialization of the data structure takes $O(n\poly(\log\Delta,1/\varepsilon))$ time.
By the bounded-spread reduction in Section~\ref{sec:decompose_bounded_spread}, we may assume $\Delta=\poly(n,1/\varepsilon)$.
We abort an execution if it exceeds the time budget $O(n\poly(\log n,1/\varepsilon))$.
Since one repetition succeeds with constant probability, repeating the procedure independently $O(\log n)$ times gives success probability at least $1-1/n^{\Omega(1)}$.
Consequently, the algorithm runs in $O(n\poly(\log n,1/\varepsilon))$ time and returns a $(1+\varepsilon)$-approximate matching with high probability.
Combining these overall time bounds and the success probability with Lemma~\ref{lem:data-structure-time-complexity} establishes our main theorem.
\main*

\section{Data Structure}\label{sec:Data-Structure}
In this section, we prove Lemma~\ref{lem:data-structure-time-complexity}.
More specifically, we maintain a data structure for $M$ in the course of the algorithm such that  one can find 
a compact symmetric augmenting path $P$ and augment $M$ along $P$
in time near-linear in the number of edges of $P$.

\subparagraph*{Directed prism graph.}
We use a directed representation of the prism graph.
The purpose of this directed representation is to reduce the search for a minimum-net-cost augmenting path to the problem
of finding a \emph{shortest path} from a free vertex in $B_0\cup B_1$ to a free vertex in $R_0\cup R_1$ containing at most one link edge.
The weight of a directed path is defined so that it coincides with the net cost of the corresponding alternating path.

We maintain two directed instances at the same time.
The first one is the \emph{upper-to-lower instance}, and the second one is the \emph{lower-to-upper instance}.
The two instances use the same underlying prism graph and the same edge weights, but only the directions of the link edges are reversed.
Thus, it is enough to describe the upper-to-lower instance; the lower-to-upper instance is defined by reversing the directions of the link edges.
Given the current matching $M$, we orient the edges of the upper-to-lower instance as follows.

If a geometric edge $e=(r,b)\in (R_0\times B_0)\cup (R_1\times B_1)$ is local with respect to $M$, then we direct it from $r$ to $b$ and assign it weight $-\Phi_M(e)$.
Otherwise, we direct it from $b$ to $r$ and assign it weight $\Phi_M(e)$.
The link edges are directed from the upper layer to the lower layer.
For the lower-to-upper instance, the link edges are directed from the lower layer to the upper layer.
The weight of the directed link edge is $\Phi_M(e)$ if it is not matched, and $-\Phi_M(e)$ otherwise.
As in~\cite{DBLP:journals/jacm/RaghvendraA20}, a directed path in this graph can be converted into a compact alternating path with the same net cost.
\begin{lemma}
\label{lem:directed-path-to-compact-path}
Any shortest directed path from a free vertex in $B_0\cup B_1$ to a free vertex in $R_0\cup R_1$ can be converted into a compact augmenting path with the same net cost.
\end{lemma}
\begin{proof}
Consider a local directed edge on the path.
If this edge already belongs to the current matching, we keep it as the corresponding matched edge of the alternating path.
Otherwise, by the definition of locality, both endpoints of this edge are incident to matched geometric edges in the same equivalence class.
We then replace the local directed edge by the length-three alternating subpath consisting of these two matched edges and the local edge between their matched partners.
Since all geometric edges within the same equivalence class share the same least common ancestor cluster and refined descendants, they have identical adjusted costs, thereby ensuring that this replacement preserves the net cost.
Applying this replacement to every local directed edge, and keeping all non-local geometric edges and link edges unchanged, yields a compact augmenting path whose net cost is equal to the cost of the directed path.
\end{proof}

From a simple observation, a symmetric augmenting path whose endpoints lie in $B_0$ and $R_1$ corresponds to a directed path in the upper-to-lower instance.
Similarly, a symmetric augmenting path whose endpoints lie in $B_1$ and $R_0$ corresponds to a directed path in the lower-to-upper instance.
Therefore, by maintaining both the upper-to-lower and lower-to-upper instances, we can search for symmetric augmenting paths by solving shortest-path problems in the two directed prism graph instances.

We do not explicitly represent all matched link edges in each directed instance.
Consider the upper-to-lower instance.
Suppose that a matched link edge has its upper endpoint in $B_0$ and its lower endpoint in $R_1$.
This edge cannot be traversed by any directed path in this instance, because its upper endpoint is not reachable through any directed geometric edge.
Indeed, no local geometric edge is incident to this endpoint, and every non-local geometric edge incident to it is directed away from it.
Thus, in the upper-to-lower instance, only matched link edges whose upper endpoint lies in $R_0$ can be reached.
The lower-to-upper instance is symmetric.
The following observation summarizes the reachable matched link edges in the two directed instances.

\begin{observation}\label{obs:matched-link-edge-one-side}
In the upper-to-lower instance, it is sufficient to represent only the matched link edges from $R_0$ to $B_1$.
In the lower-to-upper instance, it is sufficient to represent only the matched link edges from $R_1$ to $B_0$.
\end{observation}

\subparagraph*{Reduction to a vertex-weighted instance.}
Thus in the following, we focus on the upper-to-lower layer instance.
Now the goal is reduced to finding two shortest paths in the upper-to-lower instance:
(i) a shortest path from a free vertex in $B_0$ to a free vertex in $R_0$ contained in the upper layer, or (ii) 
a symmetric shortest path from a free vertex in $B_0$ to a free vertex in $R_1$.
A shortest path belonging to the first type can be computed exactly the same as in~\cite{DBLP:journals/jacm/RaghvendraA20} as it is an augmenting path fully contained in the upper layer.

Thus focus on a shortest path $\pi$ belonging to the second type, and let $e_\text{L}=(v_0,v_1)$ be the unique link edge in $\pi$ with $v_0\in R_0\cup B_0$
and $v_1\in R_1\cup B_1$.
The weight of $\pi$ is exactly $2\cdot (\ell(\pi_0)+\ell(e_\text{L})/2)$,
where $\ell(\cdot)$ denotes the length of a path (or an edge),
and $\pi_0$ denotes the maximal subpath of $\pi$ contained in the upper layer.
Here is an alternative view for our problem: we can view $\ell(e_L)/2$ as the \emph{vertex-weight} of $v_0$. Imagine that we add the copy of each point of $R_0\cup B_0$ at the same location with its corresponding vertex-weight 
while keeping the original points of $R_0\cup B_0$. Then we add all possible directed edges to $v_0$.
The weight of a path ending at a weighted point is
the sum of the weight of its edges and the weight of the endpoint.
Then computing a shortest path belonging to the second type is equivalent
to computing a shortest path in the upper layer from a free vertex in $B_0$ to any weighted point in $R_0\cup B_0$.
In this way, we can focus on the upper layer only.
In other words, we now focus on the geometric edges only.
Therefore, we can show that a slight modification of the data structure of~\cite{DBLP:journals/jacm/RaghvendraA20} is sufficient to address the vertex-weights.
By a slight abuse of notation, we write $R_0=R$ and $B_0=B$.

\subparagraph*{The data structure of Raghvendra and Agarwal.}
We briefly describe the data structure of Raghvendra and Agarwal~\cite{DBLP:journals/jacm/RaghvendraA20}, adapted to our split-tree setting\footnote{Their original construction is stated for Euclidean metrics using randomly shifted quadtrees. In doubling metrics, the same approach can be implemented by replacing the shifted quadtree with the randomized split-tree used to define the split-tree-induced distance.}, and briefly explain why it can be used for our purpose with modification.

Consider the randomized split-tree
$\{\mathcal P_0,\ldots,\mathcal P_{\delta+1}\}$ used to define the split-tree-induced
distances. Recall that this hierarchy can be viewed as a rooted tree whose root
is the unique cluster in $\mathcal P_{\delta+1}$. 
The data structure of Raghvendra and Agarwal~\cite{DBLP:journals/jacm/RaghvendraA20} works as follows.
For each cluster $A$ in the split-tree, 
recall that $\mathcal D[A]$ is a partition of the points in $A$ defined by descendants of $A$ that lie $\omega$ levels below $A$. They partition it further into \emph{entry} subclusters and \emph{exit} subclusters. 
If any subcluster contains a free vertex of $R$ (and $B$),
it is an exit (and entry) subcluster consisting only of
free vertices of $R$ (and $B$).

Then for each split-tree cluster $A$, they maintain an
edge-weighted graph $G(A)$, called the \emph{compressed graph}.
Each vertex $v$ of
$G(A)$ corresponds to a subcluster $A_v$  associated with a child of $A$. 
The directed edges of $G(A)$ are defined as follows.
For an entry subcluster $A_u$ and an exit subcluster $A_v$ contained in the same child $A'$ of $A$,
we add a directed edge from $u$ to $v$, and call this edge an \emph{interior edge}.
Its weight is the length of a shortest path, with respect to the 
adjusted edge weights in the directed prism graph, consisting of vertices of $A'$ from a point in $A_u$ to a
point in $A_v$.
For any two subclusters $A_u$ and $A_v$ coming from different children of $A$,
we add a directed edge from $u$ to $v$
if there is a directed edge from $A_u$ to $A_v$
in the directed prism graph. In this case, every vertex in $A_u$ has an outgoing edge to every vertex in $A_v$ in the directed prism graph. 
The weights of all these edges are all the same. 
We call this edge $uv$ a \emph{bridge edge}
and its weight is defined as the weight 
of the corresponding edges in the directed prism graph.
The complexity of $G(A)$ is $O(\poly(\log n,1/\varepsilon))$. 
Then they showed how to construct all compressed graphs in a bottom-up fashion.
In doing so, we compute all-pairs shortest paths in each compressed  graph $G(A)$. This can be done in near-linear time for all clusters as each compressed graph has complexity of $O(\poly(\log n,1/\varepsilon))$.

Then the query algorithm works in a top-down fashion as follows.
Starting from the root of the split-tree, it recovers a shortest path as follows. For the compressed graph $G(A)$ stored in the root,
it computes a shortest path $\pi$ in $G(A)$ between two vertices, one representing an entry subcluster containing free vertices in $B$ and one representing an exit cluster containing free vertices in $R$.
Recall that all-pairs shortest paths are already stored in the data structure, and thus we can find the shortest path $\pi$ efficiently.
An edge of $\pi$ is either a bridge edge or an interior edge. 
To recover the path represented by each interior edge, we recursively enter the child of $A$ containing the path. We 
repeat this until we reach leaves.
In this way the paths represented by the interior edges are recovered. 
For a bridge edge $uv$, we can pick any edge from a vertex in $A_u$ to a vertex in $A_v$. 
Observe that no two interior edges appear consecutively along $\pi$.
Therefore, 
we can obtain a shortest path between a free vertex in $B$ and a free vertex in $R$ in a consistent way.

\subparagraph*{Our modification.}
Now we briefly explain how to modify this data structure for our purpose.
Recall that our goal is to compute a shortest path from a free vertex in $B$ to a weighted non-free vertex in $B\cup R$. We call the weighted non-free vertices the \emph{terminal vertices}.
Recall that, for every compressed graph $G(A)$, 
we have $O(\poly(\log n,1/\varepsilon))$ vertices.
We copy a vertex $v$ of $G(A)$ whose corresponding subcluster contains a non-free vertex in $B\cup R$, and then we call each new vertex $v'$ a \emph{terminal node} of $G(A)$.
In addition to the original edges,
we add a directed edge $uv'$ to $G(A)$ for each original directed edge $uv$ of $G(A)$. 
If $uv$ is a bridge (or interior) edge, so is $uv'$.
Let $A_v$ be the subcluster represented by $v$. 
For a bridge edge $uv'$,
recall that any vertex in $A_u$ and any vertex in $A_v$ have the same split-tree-induced distance. 
Then the weight of $uv'$ is set as this distance plus the minimum weight of the vertices of $A_v$.
For an interior edge $uv'$,
let $A'$ be the child of $A$ containing both $A_u$ and $A_v$.
The weight of $uv'$ is the length of a shortest path, with respect to the 
adjusted edge weights in the directed prism graph, consisting of vertices of $A'$ from a point in $A_u$ to a \emph{weighted}
point in $A_v$.
Then the size of each compressed graph remains the same asymptotically, and the graph can be constructed in a bottom-up fashion as in~\cite{DBLP:journals/jacm/RaghvendraA20}.

Then the query algorithm works in a top-down fashion as follows.
Starting from the root of the split-tree, it recovers a shortest path as follows. For the compressed graph $G(A)$ stored in the root,
it computes a shortest path $\pi$ in $G(A)$ from a node representing an entry cluster 
to a terminal node.
By construction, $\pi$ contains a terminal vertex only at an endpoint.
We first recover the path represented by each interior edge, and then we recursively enter the child of $A$ containing the path. We 
repeat this until we reach leaves.
If we have a bridge edge $uv'$ in $\pi$ where $v'$ is a terminal node, 
we choose a vertex from $A_{v}$ with minimum weight,
and we may choose any vertex from $A_u$.
As before, no two interior edges appear consecutively along $\pi$.
Therefore, all paths represented by every edge of $\pi$ can be recovered in a consistent way.
In the following subsections, we describe the implementation of the data structure outlined above in detail.

\subsection{Hierarchical Clustering}
We now describe the hierarchical clustering for one fixed geometric layer.
The same construction is applied independently to the other geometric layer.
Both layers use the same randomized split-tree, where each copy $u_i$ is identified with the location of its original point $u$.

Let $A$ be a cluster of the split-tree.
We write $R_A$ and $B_A$ for the $R$-side and $B$-side vertices of the fixed layer whose locations lie in $A$.
Let $R_F$ and $B_F$ denote the sets of free $R$-side and $B$-side vertices under the current matching $M$.
Recall that $\mathcal D[A]$ is the family of descendants of $A$ that lie $\omega$ levels below $A$.

We next define the external refined descendants of a split-tree cluster $A$.
The purpose of this family is to classify vertices in $A$ whose matching partners lie outside $A$ according to the refined descendant that contains the partner.
Intuitively, if a point in $A$ is matched to a point outside $A$, then the partner must lie in a sibling branch of some ancestor of $A$.
Therefore, it is enough to collect the refined descendants of all such sibling branches.
This gives a bounded set of possible external locations for matching partners of vertices in $A$.
The external refined descendants $\mathcal N^*(A)$ is defined as follows.
For each ancestor $A'$ of $A$, consider the child of $A'$ that contains $A$.
For every other child $S$ of $A'$, we add the refined descendants of $S$.
Formally,
\begin{align*}
    \mathcal N^*(A)
    :=
    \bigcup_{A'\succeq A}
    \ \bigcup_{\substack{S\in \mathrm{ch}(A')\\ S\cap A=\emptyset}}
    \mathcal D[S].
\end{align*}
Here $A'\succeq A$ means that $A'$ is an ancestor of $A$, and $\mathrm{ch}(A')$ denotes the children of $A'$.
Since the split-tree has depth $O(\log\Delta)$ and each cluster has at most $2^{O(\ddim)}$ children, $|\mathcal N^*(A)| \le O(\log\Delta)\cdot 2^{O(\ddim)\cdot \omega} = \poly(\log\Delta,1/\varepsilon),$ for constant $\ddim$.

\subparagraph*{Hierarchical clustering.}
For each refined descendant $\xi\in\mathcal D[A]$, let $R_\xi:=R\cap \xi$ and $B_\xi:=B\cap \xi$.
We partition the points in $R_\xi$ and $B_\xi$ according to their matching status.
This gives four types of subclusters, namely free, internal, boundary, and link-matched subclusters.
The free subclusters are the sets $R_\xi^F:=R_F\cap R_\xi$ and $B_\xi^F:=B_F\cap B_\xi$.
The internal subclusters, consisting of points whose matching partners remain inside $A$, are
\begin{align*}
    R_\xi^I
    &:=
    \{r\in R_\xi \mid (r,b)\in M,\ b\in B_A\},&
    B_\xi^I
    &:=
    \{b\in B_\xi \mid (r,b)\in M,\ r\in R_A\}.
\end{align*}

For every external refined descendant $\eta\in\mathcal N^*(A)$, let $R_\eta:=R\cap\eta$ and $B_\eta:=B\cap\eta$.
The boundary subclusters, which are defined according to the refined external location of the matching partner, are
\begin{align*}
    R_\xi^\eta
    &:=
    \{r\in R_\xi \mid (r,b)\in M,\ b\in B_\eta\},&
    B_\xi^\eta
    &:=
    \{b\in B_\xi \mid (r,b)\in M,\ r\in R_\eta\}.
\end{align*}
Finally, we define the link-matched subclusters.
A vertex is called link-matched if it is matched via a link edge in $M$.
To represent link-matched vertices in the fixed geometric layer, we define
$R_\xi^L:=\{r\in R_\xi \mid r \text{ is matched via a link edge in } M\}$.
By the above observation~\ref{obs:matched-link-edge-one-side}, it suffices to store these link-matched states only on the $R$-side.
See Figure~\ref{fig:hierarchical_clustering}.

For each fixed refined descendant $\xi\in\mathcal D[A]$, we create a constant number of free, internal, and link-matched subclusters, and  $O(|\mathcal N^*(A)|)$ boundary subclusters.
Therefore, the total number of subclusters associated with $A$ is
$O(|\mathcal D[A]|\cdot(1+|\mathcal N^*(A)|))=\poly(\log\Delta,1/\varepsilon)$.

\begin{figure}
    \centering
    \includegraphics[width=0.6\textwidth]{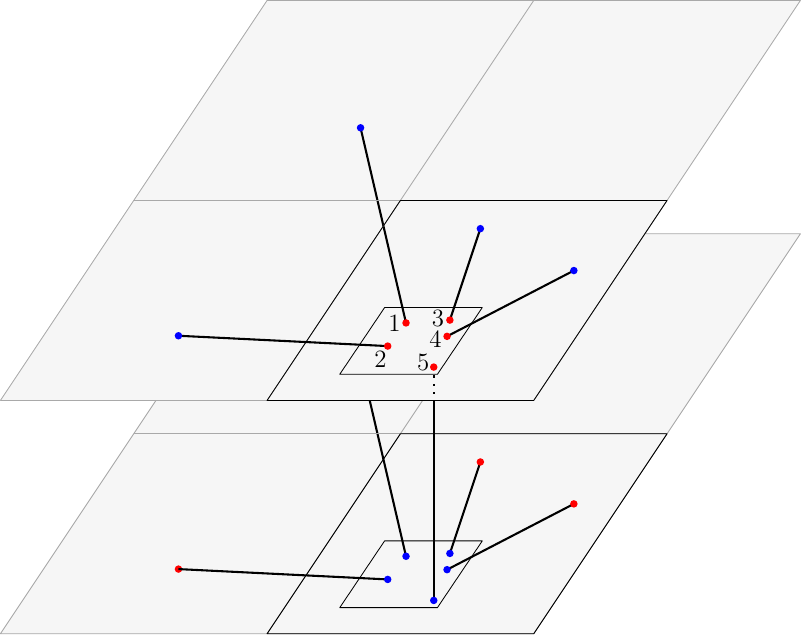}
    \caption{Hierarchical clustering for the cluster outlined in black in the upper layer. For simplicity, the figure is drawn in the Euclidean plane using quadtree cells. The singleton sets $\{1\}$ and $\{2\}$ are boundary subclusters, $\{3,4\}$ is an internal subcluster, and the singleton set $\{5\}$ is a link-matched subcluster.}
    \label{fig:hierarchical_clustering}
\end{figure}

The above clustering ensures uniform behavior of geometric edges between subclusters lying in different children of the same split-tree cluster.
Let $A_1$ and $A_2$ be two distinct children of a split-tree cluster $A$, and let $X$ and $Y$ be subclusters of $A_1$ and $A_2$, respectively.
For every $x\in X$ and $y\in Y$, the least common ancestor cluster is $A$, and the refined descendants used to define $\sdist(x,y)$ are fixed by $X$ and $Y$.
Therefore, all edges in $X\times Y$ have the same split-tree-induced distance.
Moreover, since the subclusters are defined according to the matching status and the refined location of external partners, every edge in $X\times Y$ is either local or non-local simultaneously.
Thus, all edges in $X\times Y$ have the same direction in the directed prism graph.

\subparagraph*{Hierarchical consistency.}
We next state the consistency of subclusters across consecutive levels of the split-tree.
This consistency is used to construct the compressed graph of a cluster from the information stored at its children.
Let $A_1,\ldots,A_t$ be the children of $A$ in the split-tree.
For a refined descendant $\xi\in\mathcal D[A]$, let $\mathcal D_A(\xi)$ denote the set of refined descendants of the children of $A$ that are contained in $\xi$:
\begin{align*}
    \mathcal D_A(\xi)
    :=
    \{\xi'\in \bigcup_{s=1}^t \mathcal D[A_s]\mid \xi'\subseteq \xi\}.
\end{align*}

The free, boundary, and link-matched subclusters of $A$ are obtained as unions of the corresponding child clusters:
\begin{align*}
    R_\xi^F &= \bigcup_{\xi'\in\mathcal D_A(\xi)} R_{\xi'}^F,
    &
    B_\xi^F &= \bigcup_{\xi'\in\mathcal D_A(\xi)} B_{\xi'}^F,\\
    R_\xi^\eta &= \bigcup_{\xi'\in\mathcal D_A(\xi)} R_{\xi'}^\eta,
    &
    B_\xi^\eta &= \bigcup_{\xi'\in\mathcal D_A(\xi)} B_{\xi'}^\eta,\\
    R_\xi^L &= \bigcup_{\xi'\in\mathcal D_A(\xi)} R_{\xi'}^L.
\end{align*}

The only remaining case is the internal subclusters.
A vertex in a boundary subcluster of a child $A_s$ becomes internal with respect to $A$ if its matching partner lies in another child of $A$.
For each $\xi'\in\mathcal D[A_s]$, define
$\mathcal N_A(\xi'):=\bigcup_{\substack{A_i\in\mathrm{ch}(A)\\ A_i\neq A_s}}\mathcal D[A_i].$

Then the internal subclusters are
\begin{align*}
    R_\xi^I
    &=
    \bigcup_{\xi'\in\mathcal D_A(\xi)}
    \left(
        R_{\xi'}^I
        \cup
        \bigcup_{\zeta\in\mathcal N_A(\xi')} R_{\xi'}^\zeta
    \right),
    &B_\xi^I
    =
    \bigcup_{\xi'\in\mathcal D_A(\xi)}
    \left(
        B_{\xi'}^I
        \cup
        \bigcup_{\zeta\in\mathcal N_A(\xi')} B_{\xi'}^\zeta
    \right).
\end{align*}
Thus, once the child-level subclusters have been updated, the subclusters of $A$ are determined by the formulas above.

\subsection{Compressed Graph}\label{subsec:Compressed-graph}
In this section, we construct a compressed graph for each split-tree cluster $A$.
Its vertices are the subclusters of the children of $A$.
The graph edges compactly represent directed paths inside $A$: interior edges summarize minimum-weight directed paths inside a child cluster, while bridge edges summarize single geometric transitions between different child branches.

\subparagraph*{Entry and exit subclusters.}
We define the entry subclusters of $A$ as
$\mathcal X_A^{\downarrow}:=\{\,B_\xi^F,\ R_\xi^I,\ B_\xi^\eta
\mid \xi\in\mathcal D[A],\ \eta\in\mathcal N^*(A)\,\}$.
Similarly, the exit subclusters are
$\mathcal X_A^{\uparrow}:=\{\,R_\xi^F,\ B_\xi^I,\ R_\xi^\eta,\ R_\xi^L
\mid \xi\in\mathcal D[A],\ \eta\in\mathcal N^*(A)\,\}$.
These definitions specify the possible entry and exit subclusters of a maximal directed subpath contained in $A$.
Let $\vec{\Pi}$ be a directed path of a fixed geometric layer, and let $\pi$ be a maximal connected subpath of $\vec{\Pi}$ that lies inside $A$.
If $\pi$ contains at least one geometric edge, then its first endpoint lies in $\mathcal X_A^{\downarrow}$ and its last endpoint lies in $\mathcal X_A^{\uparrow}$.
The possible entry subclusters are characterized by the first point of the maximal subpath:
\begin{enumerate}
    \item If the subpath starts at a source, then its entry point lies in a free $B$-subcluster.
    \item If the first edge of the subpath is local, then its entry point lies in an internal $R$-subcluster.
    \item If the first edge of the subpath is non-local, then its entry point lies in a boundary $B$-subcluster.
\end{enumerate}
Symmetrically, the exit point lies in a free $R$-subcluster, an internal $B$-subcluster, a boundary $R$-subcluster, or a link-matched $R$-subcluster.
Hence every maximal subpath inside $A$ is represented by a transition from an entry subcluster to an exit subcluster.

\subparagraph*{Link-crossing state.}
We also need to handle the case in which an augmenting path crosses a link edge.
For this purpose, we introduce auxiliary \emph{link-crossing states} in addition to entry and exit subclusters.
They represent the event that a directed path terminates in the current layer and then crosses a link edge to the other layer of the prism graph.

As explained in Section~\ref{sec:Data-Structure}, when the augmenting path contains a link edge $e_L=(v_0,v_1)$, the contribution of this link edge can be represented as a vertex weight of value $\ell(e_L)/2$ assigned to the vertex.
This value is $+p(v)$ if the link edge is unmatched, and $-p(v)$ if the link edge is matched.
A link-crossing state is introduced to account for this vertex contribution.
It therefore plays the same role as the terminal node used in the overview of Section~\ref{sec:Data-Structure}, and is used later as a vertex of the compressed graph.

There are two types of link-crossing states.
First, for every non-free $B$-side subcluster $Y$, we introduce a positive link-crossing state $Y^+$.
This state represents crossing an unmatched link edge at some vertex $v\in Y$.
Since the directed weight of the corresponding link edge is $+2p(v)$, the one-layer value stored at $Y^+$ includes the vertex contribution $+p(v)$.
Second, if $Y$ is a link-matched $R$-side subcluster, in particular $Y=R_\xi^L$, we introduce a negative link-crossing state $Y^-$.\footnote{In Section~\ref{sec:Data-Structure}, we describe the modification as copying all non-free vertices for simplicity. In the actual data structure, only link-matched vertices need to be copied on the $R$-side.}
This state represents crossing a matched link edge at some vertex $v\in Y$.
Since the directed weight of the corresponding link edge is $-2p(v)$, the one-layer value stored at $Y^-$ includes the vertex contribution $-p(v)$.

For each entry subcluster $X\in\mathcal X_A^\downarrow$ and each exit subcluster $Y\in\mathcal X_A^\uparrow$, we maintain a value $\psi_A(X,Y)$,
which denotes the minimum one-layer directed path weight from $X$ to $Y$ inside $A$.
Whenever the corresponding link-crossing states are defined, we also maintain $\psi_A(X,Y^+)$ and $\psi_A(X,Y^-)$, where
\begin{align*}
    \psi_A(X,Y^+)
    &:=
    \min_{v\in Y}\{\mu(X\leadsto v)+p(v)\},
    &\psi_A(X,Y^-)
    :=
    \min_{v\in Y}\{\mu(X\leadsto v)-p(v)\}.
\end{align*}
Here $\mu(X\leadsto v)$ denotes the minimum weight accumulated before crossing the link edge at $v$.
Thus, positive link-crossing states require the minimum penalty value in $Y$, while negative link-crossing states require the maximum penalty value in $Y$.
For each subcluster $Y$, we maintain vertices attaining $\min_{v\in Y}p(v)$ and $\max_{v\in Y}p(v)$, together with the corresponding penalty values.
These witnesses are used when a compressed path terminates at $Y^+$ or $Y^-$.
The factor of $2$ in the prism-graph link cost is accounted for when the resulting root-level value is doubled to obtain the weight of the corresponding symmetric alternating path.

\subparagraph*{Compressed graph.}
We now define the compressed graph associated with a split-tree cluster $A$.
Let $A_1,\ldots,A_t$ be the children of $A$.
For each child $A_s$, let $\mathcal C_{A_s}:=\mathcal X_{A_s}^{\downarrow}\cup \mathcal X_{A_s}^{\uparrow}$ be the union of entry and exit subclusters of $A_s$.
We also define the set of link-crossing states by
\begin{align*}
    \mathcal C_{A_s}^{+}
    &:=
    \{\,Y^+ \mid Y \text{ is a non-free $B$-side subcluster of } A_s\,\},\\
    \mathcal C_{A_s}^{-}
    &:=
    \{\,Y^- \mid Y \text{ is a link-matched $R$-side subcluster of } A_s\,\}.
\end{align*}
The compressed graph of $A$ is a weighted directed graph $H_A=(V_A,E_A)$ whose vertex set is
\begin{align*}
    V_A
    :=
    \bigcup_{s=1}^t
    \left(
        \mathcal C_{A_s}
        \cup
        \mathcal C_{A_s}^{+}
        \cup
        \mathcal C_{A_s}^{-}
    \right).
\end{align*}
The link-crossing states in $\mathcal C_{A_s}^{+}\cup\mathcal C_{A_s}^{-}$ have no outgoing edges.

The edge set $E_A$ consists of two kinds of directed edges.
The first kind consists of interior edges.
For each child $A_s$, every entry subcluster $X\in\mathcal X_{A_s}^{\downarrow}$, and every exit subcluster $Y\in\mathcal X_{A_s}^{\uparrow}$, we add an edge $X\to Y$ of weight $\psi_{A_s}(X,Y)$.
If the link-crossing states $Y^+$ or $Y^-$ are defined, we also add edges $X\to Y^+$ and $X\to Y^-$ with weights $\psi_{A_s}(X,Y^+)$ and $\psi_{A_s}(X,Y^-)$, respectively.
These edges summarize minimum-weight directed paths that are contained entirely inside the child $A_s$.

Second, we add bridge edges between different child branches.
Let $X\in\mathcal C_{A_s}$ and $Y\in\mathcal C_{A_{s'}}$ for two distinct children $A_s$ and $A_{s'}$.
Suppose that $X$ and $Y$ lie on opposite sides of the bipartition and that the geometric edges in $X\times Y$ are directed from $X$ to $Y$.
Since all geometric edges in $X\times Y$ have the same split-tree-induced distance and the same direction, they also have the same directed weight.
Let $x\in X$ and $y\in Y$ be arbitrary. 
We define $\beta=-\Phi_M(x,y)$ if these edges are local, and $\beta=\Phi_M(x,y)$ if these edges are non-local.
We add a bridge edge $X\to Y$ with weight $\beta$.
If $Y^+$ is defined, we also add a bridge edge $X\to Y^+$ with weight $\beta+\min_{v\in Y}p(v)$.
If $Y^-$ is defined, we add a bridge edge $X\to Y^-$ with weight $\beta-\max_{v\in Y}p(v)$.
For the edges with link-crossing states, we also store a witness vertex attaining the corresponding minimum or maximum penalty. See Figure~\ref{fig:compressed_graph}.

\begin{figure}
    \centering
    \includegraphics[width=0.81\textwidth]{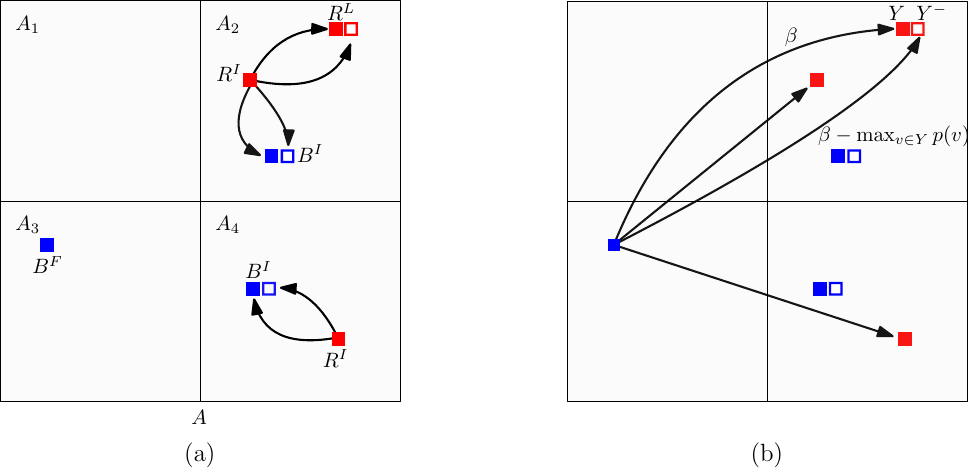}
    \caption{
    Illustration of the compressed graph $H_A$ for a split-tree cluster $A$.
    For simplicity, the figure is drawn in the Euclidean plane using quadtree cells.
    The four rectangles represent the child clusters $A_1,\ldots,A_4$ of $A$.
    Filled squares represent ordinary child-level subclusters, and open squares represent link-crossing states attached to the adjacent subclusters.
    (a) Interior edges summarize shortest-path values already computed in the compressed graph of a child cluster.
    (b) Bridge edges summarize uniform directed geometric transitions between subclusters in different child branches.
    An edge entering a link-crossing state also includes the corresponding link-crossing weight.}
    \label{fig:compressed_graph}
\end{figure}

The definition above assumes that $A$ is an internal split-tree cluster with children.
Recall that we construct a split-tree on the locations of the vertices of $X$, and two vertices of $X$ may have the same locations. 
When we say that a split-tree cluster contains a \emph{labeled vertex} $u\in X$, we mean that it contains the location of $u$.
A leaf may contain several labeled vertices with the same location.
If $A$ is a leaf, then the compressed graph is initialized directly on its nonempty subclusters, together with the minimum and maximum penalty witnesses of each subcluster.

Thus, $H_A$ compactly represents all relevant directed paths inside $A$.
The interior edges represent paths inside a child, bridge edges represent transitions between different child branches, and link-crossing states represent paths that contain link edges.

\subsection{Finding and Updating Augmenting Paths}
We now describe how the data structure finds an augmenting path and updates itself after augmentation.
The procedure is built on two basic operations, \textsc{Ascend} and \textsc{ExtractPath}.
The procedure \textsc{Ascend} constructs the compressed graph $H_A$ and computes the values $\psi_A$ using the compressed graphs, summary values, and witnesses already computed for the children of $A$.
The procedure \textsc{ExtractPath} recursively expands a path in the compressed graph into an actual directed path in the prism graph.

\subparagraph*{Ascend procedure.}
Let $A$ be an internal split-tree cluster with children $A_1,\ldots,A_t$.
Assume that the compressed graph of every child $A_s$ has already been computed, also with the values $\psi_{A_s}(X,Y)$, $\psi_{A_s}(X,Y^+)$, and $\psi_{A_s}(X,Y^-)$ for the entry, exit, and link-crossing states of $A_s$.

The procedure \textsc{Ascend}$(A)$ first constructs the compressed graph $H_A$.
The vertices of $H_A$ are the child-level subclusters and link-crossing states induced by the children of $A$.
The interior edges of $H_A$ are added using the already computed values $\psi_{A_s}$ of the children, and the bridge edges are added using the uniform directed weights between subclusters in different child branches.
It then computes shortest-path values in $H_A$ from child-level subclusters contained in an entry subcluster of $A$ to child-level subclusters contained in an exit or link-crossing state of $A$.
By hierarchical consistency, every entry subcluster $X\in\mathcal X_A^\downarrow$ and every exit subcluster $Y\in\mathcal X_A^\uparrow$ can be written as unions of subclusters from the children of $A$.
Therefore, we set $\psi_{A}(X,Y)$ as the minimum path weight in $H_A$ from any child subcluster contained in $X$ to any child subcluster contained in $Y$.

The values $\psi_A(X,Y^+)$ and $\psi_A(X,Y^-)$ are computed in the same way.
For a link-crossing state $Y^+$ or $Y^-$, the value $\psi_A(X,Y^\pm)$ is the minimum weight in $H_A$ from a child-level state contained in $X$ to a child-level link-crossing state whose underlying subcluster is contained in $Y$.
For each value computed by \textsc{Ascend}, we store the predecessor information in $H_A$ so that the corresponding compressed shortest path can later be recovered by \textsc{ExtractPath}.

\subparagraph*{ExtractPath procedure.}
The \textsc{ExtractPath} procedure reverses the computation performed by \textsc{Ascend}.
Given a split-tree cluster $A$, an entry subcluster $X$, and an exit or link-crossing state $Y$, it reconstructs an actual directed path in the prism graph.

We first consider the case where $Y$ is an ordinary exit subcluster.
The procedure recovers the corresponding shortest path in the compressed graph $H_A$ using the predecessor information stored during \textsc{Ascend}.
It then expands the edges of this compressed path in order.
If a compressed edge is an interior edge of a child $A_s$, then it corresponds to a stored $\psi$-value inside $A_s$, and the procedure recursively calls \textsc{ExtractPath} inside $A_s$.
When the recursion reaches a leaf, the stored leaf-level information specifies the actual labeled vertices used to realize the corresponding transition.
If a compressed edge is a bridge edge $X\to Y$ between two child branches, then after expanding the adjacent child-level paths, the procedure obtains an endpoint $x\in X$ on the first side and an endpoint $y\in Y$ on the second side.
It then inserts the geometric edge $(x,y)$ between the two expanded paths.
By the uniformity of inter-subcluster edges, every geometric edge in $X\times Y$ has the same direction and the same directed weight, so this inserted edge has exactly the bridge-edge weight stored in $H_A$.

When a recovered path starts or ends at a free endpoint subcluster, the actual endpoint is chosen at the leaf level.
If several free vertices at the same location are feasible, we choose a stored vertex with maximum penalty.
This choice preserves the weight of the recovered path and is best for the original penalty objective, since the chosen free vertex becomes matched after augmentation.

We next consider the case where the recovered path ends at a link-crossing state.
Suppose first that the link-crossing state is $Y^+$.
Then \textsc{ExtractPath} recovers the compressed prefix that reaches the underlying ordinary subcluster $Y$.
It chooses the stored witness in $Y$ attaining $\min_{v\in Y}p(v)$ as the endpoint of the one-layer path.
The case of $Y^-$ is analogous, except that the stored terminal witness is one attaining $\max_{v\in Y}p(v)$.
By appending the corresponding link edge and the symmetric counterpart in the other layer, we obtain a full symmetric augmenting path of weight $2\psi_A(X,Y^+)$ or $2\psi_A(X,Y^-)$.

Therefore, recursively expanding the stored information reconstructs the path represented by the selected compressed path.
If the selected path ends at an ordinary exit subcluster $Y$, the recovered one-layer directed path has weight $\psi_A(X,Y)$.
If the selected path ends at a link-crossing state $Y^+$ or $Y^-$, the resulting full symmetric augmenting path has weight $2\psi_A(X,Y^+)$ or $2\psi_A(X,Y^-)$, respectively.

\subparagraph*{Finding and updating augmenting paths.}
We now prove the data-structure guarantee stated in Lemma~\ref{lem:data-structure-time-complexity}.
\DataStructureLemma*
\begin{proof}
We maintain the upper-to-lower instance and the lower-to-upper instance simultaneously.
For the upper-to-lower instance, we maintain the one-layer data structure on the upper layer.
For the lower-to-upper instance, we maintain the one-layer data structure on the lower layer.
The two data structures are constructed and updated in the same way.

We first describe the initialization.
For each data structure, we construct the compressed graphs for all split-tree clusters in a bottom-up order.
At each internal cluster $A$, the procedure \textsc{Ascend} constructs $H_A$ from the information already computed for the children of $A$ and then computes the summary values $\psi_A$.
Since each split-tree cluster has $\poly(\log\Delta,1/\varepsilon)$ subclusters and the split-tree has $O(n\log\Delta)$ total cluster occurrences, the total initialization time is $O(n\poly(\log\Delta,1/\varepsilon))$.

At the root cluster of each data structure, we consider all stored values from free $B$-subclusters to free $R$-subclusters and to link-crossing states.
An ordinary value ending at a free $R$-subcluster represents the weight of a valid one-layer directed augmenting path.
On the other hand, a value ending at a link-crossing state stores only the one-layer contribution before adding the symmetric counterpart in the other layer.
Therefore, if a link-crossing candidate has value $\phi$, its full value is $2\phi$.
This doubling also recovers the factor $2$ in the prism-graph link weight, since the stored link-crossing value contains only $\pm p(v)$.
To handle the case in which an augmenting path immediately crosses an unmatched link edge from a free $B$-side vertex, we also include, for each free $B$-side subcluster $Y$, the direct-link candidate of full value $2\min_{v\in Y}p(v)$.
For each directed instance, we choose the minimum candidate under this comparison.
This gives one candidate from the upper-to-lower instance and one candidate from the lower-to-upper instance.
We choose the candidate with smaller full value and augment along the corresponding path.
By Lemma~\ref{lem:directed-path-to-compact-path}, the selected directed path corresponds to a compact augmenting path in the prism graph.

We now bound the update time after one augmentation.
Let $\Gamma$ be the selected compressed path.
The procedure \textsc{ExtractPath} expands $\Gamma$ into an actual one-layer directed path in the directed prism graph.
If $\Gamma$ ends at a link-crossing state, we add the corresponding link edge and the symmetric counterpart according to LSI, obtaining the full augmenting path.
Let $\Pi$ denote the resulting augmenting path.
By the construction of \textsc{ExtractPath}, the path $\Pi$ is compact and its net cost is equal to the selected full value.

After augmenting along $\Pi$, only the vertices on $\Pi$ change their matching status.
For each such vertex, we update the affected subclusters along the corresponding leaf-to-root path of the split-tree.
The witness vertices of the affected subclusters are updated together with the subclusters.
We then recompute the affected compressed graphs bottom-up using \textsc{Ascend}.
Since the split-tree has depth $O(\log\Delta)$ and each affected split-tree cluster has $\poly(\log\Delta,1/\varepsilon)$ subclusters, the total time for extracting $\Pi$ and updating the data structure is
$
|\Pi|\cdot\poly(\log\Delta,1/\varepsilon).
$
This proves Lemma~\ref{lem:data-structure-time-complexity}.
\end{proof}

\section{Decomposition into Subproblems with Bounded Spread}\label{sec:decompose_bounded_spread}
In this section, we detail how to decompose the original problem into independent subproblems defined on pseudo-metric spaces with bounded spread.
The \emph{diameter} of a set is the maximum distance between any pair of points.
To bound the probability of cutting the edges of a fixed optimal matching during the space partitioning, the diameter of the resulting subsets must be proportional to the optimal cost.
Since this cost is not inherently bounded, we scale both the metric and the penalty function to constrain the optimal cost within $\Omega(n/\varepsilon)$ and $O(n^2/\varepsilon)$.
We then map the points to a set of representative centers $Y \subseteq X$, inducing a pseudometric $d_s$ that enforces a minimum non-zero distance of at least one.
The additive error introduced by the mapping $s: X \to Y$ is bounded by the scaling.
Next, we apply a randomized low-diameter decomposition to partition the space into disjoint subsets.
This preserves an optimal matching within the subsets with probability at least $1/2$.
Consequently, solving the subproblems independently yields a $(1+O(\varepsilon))$-approximate solution.
The \emph{spread} of a set is the ratio of its diameter to its minimum non-zero distance.
The decomposition allows the assumption that the input instance has a bounded spread of $O(n^2/\varepsilon)$, which restricts the depth of hierarchical data structures and enables a near-linear running time.

\subsection{Metric Perturbation via Scaling}\label{subsec:perturbation_via_scaling}
We first scale the metric and the penalty function to bound the optimal cost in terms of $n$, and perturb the points to define a pseudometric $d_s$ that has a minimum non-zero distance of at least one.
We compute an estimate $T$ to scale $d$ and $p$ by $\alpha = O(n / (\varepsilon T))$, normalizing the optimal cost within $\Omega(n/\varepsilon)$ and  $O(n^2/\varepsilon)$.
We then select a set of centers $Y \subseteq X$ under the scaled metric and define a mapping $s \colon X \to Y$ to induce $d_s$, ensuring that the distance between any two distinct centers is strictly greater than one.
The scaling absorbs the $O(n)$ additive error from the mapping into a relative error, preserving a $(1+O(\varepsilon))$-approximation.

\medskip

To obtain an estimate $T$ of the optimal cost under $d$, we first compute a sparse graph $G_2$ to approximate the metric space $(X, d)$.
We then construct a tree metric $d_H$ from $G_2$ and compute the corresponding optimal cost.
The estimate satisfies $T \le \mathrm{w}(M^*) \le 2(n-1) \cdot T$, where $M^*$ is an optimal matching under $d$.

\subparagraph*{Computing a $2$-spanner of the metric space.}
Given a parameter $t \ge 1$, a \emph{$t$-spanner} of $(X, d)$ is a weighted graph $G_t$ defined on $X$.
The weight of an edge in $G_t$ is the distance between its endpoints under $d$.
For any pair of points $u, v \in X$, the shortest path distance $d_{G_t}(u, v)$ in the graph $G_t$ satisfies $d(u, v) \le d_{G_t}(u, v) \le t \cdot d(u, v)$.
Gottlieb and Roditty~\cite{DBLP:conf/esa/GottliebR08} introduced a dynamic data structure that maintains a $t$-spanner with a constant maximum degree in metric spaces with constant doubling dimension.
This data structure supports a point insertion in logarithmic time.
Using this result, we compute a $2$-spanner $G_2$ of the metric space $(X, d)$ with a constant maximum degree in $O(n \log n)$ time.

\subparagraph*{Construction of a tree metric.}
We define a \emph{hierarchically well-separated tree} (HST) as a metric space over the leaves of a tree.
\begin{definition}\label{def:hst}
A hierarchically well-separated tree (HST) is a metric space defined on the leaves of a rooted tree. 
Each node $v$ in the tree is associated with a label $\Delta_v \ge 0$ such that $\Delta_v = 0$ if and only if $v$ is a leaf, and $\Delta_u \le \Delta_v$ for any child $u$ of $v$. 
The distance between any two leaves is the label of their lowest common ancestor.
\end{definition}
Without loss of generality, we assume that the underlying tree of an HST is a binary tree~\cite{DBLP:journals/siamcomp/Har-PeledM06}.
The following lemma guarantees that an HST over the vertices of a weighted connected graph can be constructed in near-linear time with respect to the number of its vertices and edges.
\begin{lemma}[\cite{DBLP:journals/siamcomp/Har-PeledM06}]\label{lem:graph_to_hst}
For a weighted connected graph $G' = (V', E')$ with $n'$ vertices and $m'$ edges, we can compute an HST $H'$ on $V'$ in $O(n' \log n' + m')$ time.
The tree metric $d_{H'}$ of $H'$ satisfies $d_{G'}(u', v') \le d_{H'}(u', v') \le (n' - 1) \cdot d_{G'}(u', v')$ for all $u', v' \in V'$, where $d_{G'}$ is the shortest path metric of $G'$.
\end{lemma}
By applying Lemma \ref{lem:graph_to_hst} to $G_2$, we compute an HST $H$ on $X$ whose tree metric $d_H$ guarantees an $(n-1)$-approximation of the shortest path metric of $G_2$.

Given a matching $M \subseteq R \times B$, we define its cost under the tree metric $d_H$ as $\mathrm{w}_H(M) = \sum_{(r,b) \in M} d_H(r, b) + \sum_{v \in U(M)} p(v)$. 
We establish the following lemma to bound the cost of a matching under $d_H$ in terms of the original cost under $d$.
\begin{lemma}\label{lem:matching_cost_hst}
For every matching $M \subseteq R \times B$, $\mathrm{w}(M) \le \mathrm{w}_H(M) \le 2(n - 1) \cdot \mathrm{w}(M)$.
\end{lemma}
\begin{proof}
Since $d_H$ provides an $(n-1)$-approximation of the shortest path metric of the spanner $G_2$ of $(X, d)$, the distance function satisfies $d(r, b) \le d_H(r, b) \le 2(n - 1) \cdot d(r, b)$ for all $(r, b) \in R \times B$.
Applying the lower bound $d(r, b) \le d_H(r, b)$ to each edge in $M$ guarantees $\mathrm{w}(M) \le \mathrm{w}_H(M)$.
Similarly, applying $d_H(r, b) \le 2(n - 1) \cdot d(r, b)$ to each edge in $M$ gives $\mathrm{w}_H(M) \le 2(n - 1) \sum_{(r,b) \in M} d(r, b) + \sum_{v \in U(M)} p(v)$.
Bounding the penalty term by $2(n - 1) \sum_{v \in U(M)} p(v)$ yields $\mathrm{w}_H(M) \le 2(n - 1) \cdot \mathrm{w}(M)$.
\end{proof}

\subparagraph*{Estimation of the optimal cost.}
We now compute the optimal cost under the tree metric $d_H$ to establish the estimate $T$.
Sato et al.~\cite{DBLP:conf/nips/SatoYK20} presented a quasi-linear time exact algorithm for generalized unbalanced optimal transport problems on tree metrics.
As detailed in Appendix~\ref{sec:penaly_matching_tree}, by adapting their algorithm to our setting, we can exactly compute the optimal cost $W^*_H$ under $d_H$ in $O(n \log^2 n)$ time.
Since $d(r, b) \le d_H(r, b)$ for all $(r, b) \in R \times B$, Lemma \ref{lem:matching_cost_hst} implies $\mathrm{w}(M^*) \le W^*_H \le 2(n - 1) \cdot \mathrm{w}(M^*)$, where $M^*$ is an optimal matching under $d$.
Therefore, setting the estimate $T = W^*_H/(2(n - 1))$ yields $T \le \mathrm{w}(M^*) \le 2(n-1) \cdot T$.
We assume $T > 0$.

\medskip

Using the estimate $T$, we define a scaling factor $\alpha = C \cdot \frac{n}{\varepsilon T}$, where $C$ is a sufficiently large constant.
The scaled distance function $d_{\alpha}$ and the scaled penalty function $p_{\alpha}$ are defined as $d_{\alpha}(u, v) = \alpha \cdot d(u, v)$ for all $u, v \in X$ and $p_{\alpha}(x) = \alpha \cdot p(x)$ for all $x \in X$.
We now consider the metric space $(X, d_{\alpha})$ with the penalty function $p_{\alpha}$, where the optimal cost is guaranteed to be bounded between $\Omega(n/\varepsilon)$ and $O(n^2/\varepsilon)$ due to the bounds on $T$.

\medskip

Within the metric space $(X, d_\alpha)$, we compute a set of representative centers $Y \subseteq X$ and define a mapping $s \colon X \to Y$.
For a parameter $r > 0$ and a \emph{packing constant} $\beta \ge 1$, we define an \emph{$r$-net} of a metric space as follows.
\begin{definition}\label{def:r_net}
Given a metric space $(X', d')$, a parameter $r > 0$, and a packing constant $\beta \ge 1$, an $r$-net of $X'$ is a subset $Y'$ of $X'$ such that $\min_{u, v \in Y', u \neq v} d'(u, v) \ge r/\beta$ and $\max_{x \in X'} d'(x, Y') \le r$, where $d'(x, Y') = \min_{y \in Y'} d'(x, y)$.
\end{definition}
For a constant $\kappa > 1$, we compute a $\kappa$-net $Y$ of $X$ under the metric $d_\alpha$ using a packing constant $\beta = \kappa$, which guarantees that $d_\alpha(u, v) > 1$ for any distinct points $u, v \in Y$.
The mapping $s \colon X \to Y$ is defined during the computation.

\subparagraph*{Computation of the $\kappa$-net and the mapping.}
We initialize $Y$ as an empty set.
Cole and Gottlieb~\cite{DBLP:conf/stoc/ColeG06} introduced a dynamic data structure for metric spaces with constant doubling dimension that supports approximate nearest neighbor queries and point insertions in logarithmic time.
This data structure $D$ is maintained over $Y$ to answer $\kappa$-approximate nearest neighbor queries under the metric $d_\alpha$.
For each point $v \in X$, if $Y$ is empty, we add $v$ to $Y$, insert $v$ into $D$, and set $s(v) = v$.
Otherwise, we query $D$ to find a $\kappa$-approximate nearest neighbor $q \in Y$ of $v$.
This query guarantees that $d_\alpha(v, Y) \le d_\alpha(v, q) \le \kappa \cdot d_\alpha(v, Y)$.
If $d_\alpha(v, q) > \kappa$, we add $v$ to $Y$, insert $v$ into $D$, and set $s(v) = v$.
If $d_\alpha(v, q) \le \kappa$, we do not add $v$ to $Y$ and simply set $s(v) = q$.
Since querying and updating $D$ take logarithmic time, the entire procedure takes $O(n \log n)$ time.

To prove that $Y$ is a $\kappa$-net of $X$ under $d_\alpha$ with a packing constant $\kappa$, we show that $\min_{u, v \in Y, u \neq v} d_\alpha(u, v) > 1$ and $\max_{x \in X} d_\alpha(x, Y) \le \kappa$.
The construction of $Y$ directly implies $\max_{x \in X} d_\alpha(x, Y) \le \kappa$.
To verify that distinct points in $Y$ are separated by a distance of at least $1$, consider an arbitrary point $y$ added to $Y$.
Let $Y_{\textsf{old}}$ denote the set of points in $Y$ existing prior to the insertion of $y$.
The insertion rule guarantees $\kappa < d_\alpha(y, q) \le \kappa \cdot d_\alpha(y, Y_{\textsf{old}})$ for any $q \in Y_{\textsf{old}}$, which implies $d_\alpha(y, Y_{\textsf{old}}) > 1$.
Therefore, $Y$ is a $\kappa$-net of $X$ under $d_\alpha$ satisfying $\min_{u, v \in Y, u \neq v} d_\alpha(u, v) > 1$.

\medskip

We now define a pseudometric $d_s$ on $R \times B$ and bound the cost of an optimal matching under $d_s$ with respect to the original metric $d$.
For all $(r, b) \in R \times B$, we define $d_s(r, b) = d_\alpha(s(r), s(b))$.
For a matching $M \subseteq R \times B$, its cost under $d_s$ is defined as $\mathrm{w}_s(M) = \sum_{(r,b) \in M} d_s(r, b) + \sum_{v \in U(M)} p_\alpha(v)$.
The following lemma shows that an optimal matching under $d_s$ yields a valid approximation for the original problem.
\begin{lemma}\label{lem:approx_bound_ds}
Let $M^*$ and $M^*_s$ be the minimum-cost bipartite matchings with penalties between $R$ and $B$ under $d$ and $d_s$, respectively.
Then, $\mathrm{w}(M^*) \le \mathrm{w}(M^*_s) \le (1 + O(\varepsilon)) \cdot \mathrm{w}(M^*)$.
\end{lemma}
\begin{proof}
As $M^*$ is an optimal matching under $d$, $\mathrm{w}(M^*) \le \mathrm{w}(M^*_s)$ trivially holds.
By the definition of the mapping $s$, $d_\alpha(v, s(v)) \le \kappa$ for all $v \in X$.
For any pair $(r, b) \in R \times B$, the triangle inequality yields $|d_s(r, b) - d_\alpha(r, b)| \le 2\kappa$.
Since the penalty terms are identical under $d_s$ and $d_\alpha$, the cost difference for any matching $M \subseteq R \times B$ is bounded by 
$|\mathrm{w}_s(M) - \alpha \cdot \mathrm{w}(M)| \le \sum_{(r,b) \in M} |d_s(r, b) - d_\alpha(r, b)| \le 2\kappa |M| \le \kappa n$.

The optimality of $M^*_s$ under $d_s$ implies $\mathrm{w}_s(M^*_s) \le \mathrm{w}_s(M^*)$.
Applying the cost difference bound to both $M^*$ and $M^*_s$ yields
$\alpha \cdot \mathrm{w}(M^*_s) \le \mathrm{w}_s(M^*_s) + \kappa n \le \mathrm{w}_s(M^*) + \kappa n \le \alpha \cdot \mathrm{w}(M^*) + 2\kappa n$.
Dividing by $\alpha = Cn/\varepsilon T$ results in $\mathrm{w}(M^*_s) \le \mathrm{w}(M^*) + \frac{2\kappa}{C} \varepsilon \cdot T$.
Since $T \le \mathrm{w}(M^*)$, we conclude that $\mathrm{w}(M^*) \le \mathrm{w}(M^*_s) \le (1 + \frac{2\kappa}{C} \varepsilon) \cdot \mathrm{w}(M^*) = (1 + O(\varepsilon)) \cdot \mathrm{w}(M^*)$.
\end{proof}
By Lemma~\ref{lem:approx_bound_ds}, finding an optimal matching under $d_s$ guarantees a $(1+O(\varepsilon))$-approximate solution for the original problem.
Under the pseudometric $d_s$, the minimum non-zero distance is at least one and the optimal cost is bounded in terms of $n$.
This structure ensures that applying a low-diameter decomposition partitions the instance into independent subproblems with a bounded spread of $O(n^2/\varepsilon)$.

\subsection{Randomized Low-Diameter Decomposition}\label{subsec:randomized_ldd}
We now partition the point set $X$ into disjoint clusters of diameter of $O(n^2/\varepsilon)$ under $d_\alpha$ while preserving an optimal matching $M^*_s$ under the pseudometric $d_s$ with probability at least $1/2$.
We first select a subset of centers $Z \subseteq Y$ that forms an $(L/4)$-net under $d_\alpha$ for a radius bound $L = O(n^2/\varepsilon)$.
Following the randomized low-diameter decomposition scheme of~\cite{DBLP:conf/stoc/Talwar04}, we assign each point of $Y$ to a center in $Z$, partitioning $Y$ into disjoint subsets with a bounded spread of $O(L)$ under $d_\alpha$.
Through the mapping $s \colon X \to Y$, this assignment induces a partition of $X$ into disjoint clusters, ensuring that each cluster has a bounded diameter of $O(L)$ under $d_\alpha$.
Independently computing an optimal matching under $d_s$ within each cluster yields a $(1+O(\varepsilon))$-approximate solution for the original problem.

\medskip

We detail the procedure to compute the randomized decomposition.
We first construct the $(L/4)$-net $Z$ of $Y$ with a packing constant $\chi > 1$.
We then assign each point in $Y$ to a center in $Z$ based on a random radius $\ell$ and a random permutation $\pi$ of $Z$.
Relying on the properties of the bounded doubling dimension, this procedure runs in $O(n \log n)$ time using the dynamic data structure of~\cite{DBLP:conf/stoc/ColeG06}.

\subparagraph*{Computation of the randomized decomposition.}
We initialize $Z = \emptyset$ and maintain a dynamic data structure $D_1$~\cite{DBLP:conf/stoc/ColeG06} over $Z$ to answer $\chi$-approximate nearest neighbor queries under $d_\alpha$.
For each $y \in Y$, if $Z = \emptyset$, we insert $y$ into $Z$ and $D_1$.
Otherwise, we query $D_1$ for a $\chi$-approximate nearest neighbor $q \in Z$ of $y$.
If $d_\alpha(y, q) > L/4$, we insert $y$ into $Z$ and $D_1$.

To partition $Y$, we sample $\rho \in [1/2, 1)$ uniformly at random and set $\ell = \rho L$.
We generate a permutation $\pi$ of $Z$ uniformly at random.
We define a mapping $t \colon Y \to Z$ using a dynamic data structure $D_2$~\cite{DBLP:conf/stoc/ColeG06} initialized with $Z$.
For each $y \in Y$, we initialize $Z_y = \emptyset$ and repeatedly query $D_2$ for a $\chi$-approximate nearest neighbor $q$ of $y$.
If $d_\alpha(y, q) \le \chi \ell$, we add $q$ to $Z_y$ and remove it from $D_2$.
Otherwise, we terminate the queries for $y$.
This ensures that $d_\alpha(y, z) > \ell$ for all centers $z$ remaining in $D_2$, guaranteeing that $Z_y$ contains every center $z \in Z$ satisfying $d_\alpha(y, z) \le \ell$.
We set $t(y)$ to the center among $\{z \in Z_y \mid d_\alpha(y, z) \le \ell\}$ that appears first in $\pi$, and we reinsert all points of $Z_y$ into $D_2$.

For each $z \in Z$, we define a disjoint cluster $X_z = \{v \in X \mid t(s(v)) = z\}$.
We also let $R_z = X_z \cap R$ and $B_z = X_z \cap B$.
The clusters $X_z$ for all $z \in Z$ form the partition of $X$.
By construction, the diameter of $X_z$ and the spread of $X_z \cap Y$ under $d_\alpha$ are both bounded by $O(L)$.

\subparagraph*{Correctness and time complexity.}
As in Section~\ref{subsec:perturbation_via_scaling}, $Z$ forms an $(L/4)$-net of $Y$ under $d_\alpha$ with a packing constant $\chi$.
Computing $Z$ takes $O(n \log n)$ time using the dynamic data structure.

To bound the time complexity of defining the mapping $t$, we rely on a fundamental property of doubling metric spaces.
\begin{lemma}[\cite{DBLP:conf/focs/GuptaKL03}]\label{lem:doubling_aspect_ratio}
Let $(X', d')$ be a doubling metric space with doubling dimension $\ddim$.
If a subset $Y' \subseteq X'$ has a spread of $\Lambda$, then $|Y'| \le 2^{\ddim \cdot \lceil \log_2 \Lambda \rceil}$.
\end{lemma}
Defining $t$ requires repeatedly querying $D_2$ to compute $Z_y$ for each $y \in Y$.
The number of queries per point $y$ is exactly $|Z_y| + 1$.
Let $B_{d_\alpha}(y, r)$ denote the closed ball of radius $r$ centered at $y$ under $d_\alpha$.
The following lemma bounds $|Z_y|$.
\begin{lemma}\label{lem:candidate_set_size}
For every $y \in Y$, $|Z_y| = O(1)$.
\end{lemma}
\begin{proof}
By construction, every center $q \in Z_y$ satisfies $d_\alpha(y, q) \le \chi \ell$.
Since $\rho < 1$, we have $\ell < L$, which implies $Z_y \subseteq Z \cap B_{d_\alpha}(y, \chi L)$.

As $Z$ is an $(L/4)$-net with a packing constant $\chi$, the minimum distance between distinct centers in $Z$ is at least $L / (4\chi)$.
The diameter of $Z \cap B_{d_\alpha}(y, \chi L)$ is at most $2\chi L$.
Thus, the spread of $Z \cap B_{d_\alpha}(y, \chi L)$ is bounded by $8\chi^2$.
By Lemma~\ref{lem:doubling_aspect_ratio}, we obtain $|Z \cap B_{d_\alpha}(y, \chi L)| \le 2^{\ddim \lceil \log_2 (8\chi^2) \rceil}$.
Since $\ddim$ and $\chi$ are constants, we conclude $|Z_y| \le |Z \cap B_{d_\alpha}(y, \chi L)| = O(1)$.
\end{proof}
Since $|Z_y| = O(1)$, computing $Z_y$ requires $O(1)$ queries for each point $y \in Y$.
Therefore, the entire decomposition procedure runs in $O(n \log n)$ time.

\medskip

We determine the radius bound $L$ that preserves an optimal matching $M^*_s$ under $d_s$ within the clusters with probability at least $1/2$.
We first bound the probability that the mapping $t$ assigns two points in $Y$ to different centers.

\begin{lemma}\label{lem:prob_separation}
For any pair of points $y, y' \in Y$, $\Pr[t(y) \neq t(y')] \le \gamma \cdot \frac{d_\alpha(y, y')}{L}$, where $\gamma = O(\ddim \cdot \log \chi)$.
\end{lemma}
\begin{proof}
Assume $d_\alpha(y, y') \le L / (2\gamma)$.
Otherwise, the bound is trivial.
Order the centers $z \in Z$ as $z_1, z_2, \dots$ in non-decreasing order of $\min(d_\alpha(y, z), d_\alpha(y', z))$.
The event $t(y) \neq t(y')$ occurs only if some center $z_j$ covers exactly one of $y$ and $y'$ under the radius $\ell$, and it appears before all other centers covering either $y$ or $y'$ in the random permutation $\pi$.

For $z_j$ to cover exactly one point, $\ell$ must fall between $d_\alpha(y, z_j)$ and $d_\alpha(y', z_j)$.
By the triangle inequality, the length of this interval is at most $d_\alpha(y, y')$.
Since $\ell$ is chosen uniformly at random from $[L/2, L)$, this occurs with probability at most $2 d_\alpha(y, y') / L$.

Conditioned on this radius event, $z_j$ dictates the assignment only if it appears before all other covering centers in $\pi$.
As the centers are ordered by the minimum distance to the pair, $\ell \ge \min(d_\alpha(y, z_j), d_\alpha(y', z_j))$ guarantees that $z_1, \dots, z_{j-1}$ also cover at least one point of the pair.
Thus, it must appear before $z_1, \dots, z_{j-1}$.
Since $\pi$ is chosen uniformly at random and is independent of $\ell$, this occurs with probability of $1/j$.

By the union bound over all centers, $\Pr[t(y) \neq t(y')] \le \frac{2 d_\alpha(y, y')}{L} \sum_{j=1}^K \frac{1}{j} \le \frac{2 d_\alpha(y, y')}{L} \cdot (\ln K + 1)$, where $K$ is the number of centers capable of covering at least one point of the pair.
Any such center lies within a distance of $L + d_\alpha(y, y') \le 2L$ from $y$.
Since $Z$ is an $(L/4)$-net with a packing constant $\chi$, the minimum distance between distinct centers in $Z$ is at least $L/(4\chi)$.
Thus, these $K$ centers form a subset with a spread of at most $\frac{4L}{L/(4\chi)} = 16\chi$ under $d_\alpha$.
By Lemma~\ref{lem:doubling_aspect_ratio}, $K \le 2^{\ddim \lceil \log_2 (16\chi) \rceil}$, bounding the harmonic sum by $O(\ddim \cdot \log \chi)$.
Setting $\gamma = O(\ddim \cdot \log \chi)$ concludes the proof.
\end{proof}
An edge $(r, b) \in M^*_s$ is cut by the partition if $t(s(r)) \neq t(s(b))$.
By Lemma~\ref{lem:prob_separation}, the probability that the partition cuts any edge of $M^*_s$ is at most $\sum_{(r,b) \in M^*_s} \gamma \cdot \frac{d_s(r, b)}{L} \le \gamma \cdot \frac{\mathrm{w}_s(M^*_s)}{L}$.
Since $\mathrm{w}_s(M^*_s) \le 2Cn(n-1) / \varepsilon + \kappa n$, setting $L = O(n^2/\varepsilon)$ guarantees that $M^*_s$ is completely preserved within the clusters with probability at least $1/2$.

\medskip

The decomposition reduces the original problem to computing an optimal matching independently on each cluster, which yields a $(1+O(\varepsilon))$-approximate solution with probability at least $1/2$.
By construction, each cluster has a bounded diameter of $O(n^2/\varepsilon)$ under $d_\alpha$, and the pseudometric $d_s$ enforces a minimum non-zero distance of at least one.

To connect this decomposition to the setting analyzed in the preceding sections, we focus on a single cluster as our problem instance.
Let $X = R \cup B$ be the point set of the cluster.
We redefine the distance function $d$ to be the pseudometric $d_s$, and the penalty function $p$ to be the scaled penalty $p_\alpha$.
Under this new distance function, distinct points in $X$ may have a distance of zero. We say that they have the same \emph{location}.
We maintain these co-located points as distinct entities to preserve their individual penalties.
This formalizes our earlier assumption that the input $X$ with $d$ and $p$ has a minimum non-zero distance of at least one and a bounded spread of $O(n^2/\varepsilon)$.



\bibliography{lipics-v2021-sample-article}

\newpage

\appendix

\section{Penalty Matching on a Tree}\label{sec:penaly_matching_tree}
In this section, we present an exact algorithm for the minimum-cost bipartite matching with penalties when the underlying metric is a tree metric.
This algorithm will be used as a subroutine for the perturbation step.
After constructing an HST for the input point set, we solve the penalty matching instance on a tree exactly and use its value as a coarse estimate of the optimum value of the original instance.
For simplicity, we regard the HST as an edge-weighted binary tree.
This assumption is without loss of generality, since any HST can be transformed into an edge-weighted binary tree without changing the induced metric on the leaves.

\subsection{Problem setting}
Let $\mathcal{T}=(V,E,w)$ be a weighted binary tree.
For a node $v$, let $\mathcal{T}_v$ denote the subtree rooted at $v$.
The distance between $x$ and $y$ is denoted by $d_{\mathcal T}$, where $d_{\mathcal T}(x,y)$ is the total weight of the unique path between $x$ and $y$ on $\mathcal{T}$.
For a non-root node $v\in V$, let $\operatorname{par}(v)$ denote the parent of $v$.
We assume that every point lies on a leaf, and that internal nodes do not contain input points.
Each point $q\in R\cup B$ has a nonnegative penalty $p(q)$.

A matching is a set $M\subseteq R\times B$ such that each point in $R\cup B$ is incident to at most one pair in $M$.
Let $U(M)\subseteq R\cup B$ be the set of unmatched points. The cost of $M$ on
$\mathcal{T}$ is defined as
\begin{align*}
\operatorname{cost}_{\mathcal{T}}(M)
=
\sum_{(r,b)\in M} d_{\mathcal{T}}(r,b)
+
\sum_{q\in U(M)} p(q).
\end{align*}
The goal is to compute a matching of minimum cost.

\subsection{Algorithm}
We compute the optimum value by a bottom-up dynamic program on the tree.
For an integer $k$, we define $F_v(k)$ as the minimum cost of a partial solution inside $\mathcal{T}_v$ with the imbalance $k$.
Here, imbalance $k$ means the $|k|$ points that remain exposed at the root of $\mathcal{T}_v$ after all
decisions inside the subtree are made.
A positive value $k>0$ means that $k$ red points in $\mathcal{T}_v$ remain unmatched inside $\mathcal{T}_v$ and must be matched to blue points outside $\mathcal{T}_v$. A negative value $k<0$ means that $-k$ blue points in $\mathcal{T}_v$ remain unmatched inside $\mathcal{T}_v$ and must be matched to red points outside $\mathcal{T}_v$. The case $k=0$ means that no point is exposed outside the subtree.

The following uncrossing lemma justifies why a single signed imbalance $k$ is sufficient.

\begin{lemma}
    For every node $v$, there exists an optimal solution in which the points of $\mathcal{T}_v$ matched outside $\mathcal{T}_v$ have only one color.
\end{lemma}
\begin{proof}
    Let $e=(v,\operatorname{par}(v))$ be the edge connecting $\mathcal{T}_v$ to the rest of the tree.
    Suppose that a matching contains two crossing pairs $(r_1,b_2)$ and $(r_2,b_1)$ such that
    $r_1,b_1\in \mathcal{T}_v$ and $r_2,b_2\notin \mathcal{T}_v$.
    We show that replacing these two pairs by $(r_1,b_1)$ and $(r_2,b_2)$ does not increase the cost.
    Note that
    $d_{\mathcal{T}}(r_1,b_2)= d_{\mathcal{T}}(r_1,v) +w(e)+ d_{\mathcal{T}}(\operatorname{par}(v),b_2)$.
    Similarly,
    $d_{\mathcal{T}}(r_2,b_1)=d_{\mathcal{T}}(r_2,\operatorname{par}(v))+w(e)+d_{\mathcal{T}}(v,b_1)$.
    
    On the other hand, by the triangle inequality inside the two components, we have
    \begin{align*}
    d_{\mathcal{T}}(r_1,b_1)
    \le d_{\mathcal{T}}(r_1,v)+d_{\mathcal{T}}(v,b_1),\qquad
    d_{\mathcal{T}}(r_2,b_2)
    \le d_{\mathcal{T}}(r_2,\operatorname{par}(v))
    +d_{\mathcal{T}}(\operatorname{par}(v),b_2).
    \end{align*}
    
    Combining the above inequalities gives
    \begin{align*}
    d_{\mathcal{T}}(r_1,b_1)+d_{\mathcal{T}}(r_2,b_2)
    &\le d_{\mathcal{T}}(r_1,v)+d_{\mathcal{T}}(v,b_1)
    +d_{\mathcal{T}}(r_2,\operatorname{par}(v))
    +d_{\mathcal{T}}(\operatorname{par}(v),b_2)\\
    &= d_{\mathcal{T}}(r_1,b_2)+d_{\mathcal{T}}(r_2,b_1)-2w(e)\\
    &\le d_{\mathcal{T}}(r_1,b_2)+d_{\mathcal{T}}(r_2,b_1).
    \end{align*}
    Thus, the total cost does not increase.

    Consequently, for each subtree, it is sufficient to record only the signed imbalance.
\end{proof}

We now describe the dynamic program. If $v$ is a leaf containing a red point $r$, then we set $F_v(0)=p(r)$ and $F_v(1)=0$, and $F_v(k)=+\infty$ for all other values of $k$. If $v$ is a leaf containing a blue point $b$, then we set $F_v(0)=p(b)$ and $F_v(-1)=0$, and $F_v(k)=+\infty$ for all other values of $k$.
For a non-root node $v$, we define the extended function $\widehat F_v(k)=F_v(k)+|k|\cdot w(v,\operatorname{par}(v))$. The term $|k|\cdot w(v,\operatorname{par}(v))$ is the cost of moving the exposed imbalance from $v$ to its parent.
    
Let $u_1$ and $u_2$ be the children of an internal node $v$. Then the dynamic programming recurrence is
\begin{align*}
F_v(k)=\min_{k_1+k_2=k}\left\{\widehat F_{u_1}(k_1)+\widehat F_{u_2}(k_2)\right\}.
\end{align*}
Equivalently, $F_v=\widehat F_{u_1}\otimes \widehat F_{u_2}$, where the min-sum convolution is defined by $(F\otimes G)(k)=\min_{i+j=k}\{F(i)+G(j)\}$. Finally, if $\rho$ is the root of $\mathcal{T}$, then the optimum value is $F_\rho(0)$.

The recurrence above is not evaluated by enumerating all possible pairs $(k_1,k_2)$. 
Instead, we use the fact that every function $F_v$ is discrete convex.
For a function $F$, define its slope at $k$ by
$s_k(F)=F(k+1)-F(k)$.
Discrete convexity means that the slopes are nondecreasing.
    
The base functions at the leaves are discrete convex.
The edge extension $F_v(k)\mapsto F_v(k)+|k|w(v,\operatorname{par}(v))$ also preserves discrete convexity. 
Moreover, the min-sum convolution of two discrete convex functions is again discrete convex, and the slope list of the convolution is obtained by merging the two slope lists in nondecreasing order.
Therefore, each merge step can be performed by slope-list merging rather than by enumerating all imbalance splits.
    
\subsection{Time Complexity}

In this section, we explain how the running time can be made near-linear.
We represent each discrete convex function by its minimum value, a minimizer, and its slope list.
The slope list is stored in a balanced binary search tree, and each slope element corresponds to one node of the tree.
Each node stores its slope value and a lazy offset that is applied to all slope elements in its subtree.
The lazy offset is used for the edge extension $F(k)\mapsto F(k)+|k|w(v,\operatorname{par}(v))$.
This operation changes the slopes by adding uniform offsets to the relevant slope elements.
When the same value has to be added to all slopes in a subtree, we update only the lazy offset of the subtree root.
    
Since each input point contributes one slope element, the total number of slope elements is $O(n)$, where $n=|R|+|B|$.
When two child functions are merged, we merge their slope lists.
We use a small-to-large merging strategy in which all slope elements of the smaller list are inserted into the larger list.
During the search, we accumulate and push the lazy offsets along the search path, compare the resulting actual slope values, and insert the new slope element at the correct position.
After the insertion, we update the lazy offsets of the affected subtrees.
Therefore, each insertion into the balanced binary search tree takes $O(\log n)$ time.
Under the small-to-large strategy, each slope element is inserted into a larger list only $O(\log n)$ times over the entire bottom-up computation.
Thus, all merge operations take $O(n\log^2 n)$ time.
The edge-extension operations are dominated by the merge operations.
Therefore, the total running time is $O(n\log^2 n)$.
\end{document}